\documentclass[onecolumn,authoryear]{els-mrw} 

\usepackage{amsmath,amssymb,amsfonts,amsthm,makeidx,graphicx}
\usepackage{txfonts}
\usepackage{helvet}

\def\Sec#1{Sec.~\ref{#1}}
\newcommand{\ud}{{\mathrm d}}
\renewcommand{\Re}{{\mathbb R}}

\renewcommand{\k}{{\mathsf K}}
\usepackage{caption}
\usepackage{subcaption}
\usepackage{tikz}
\usetikzlibrary{shapes,arrows,positioning,calc}
\usetikzlibrary{shapes.geometric,arrows}

\tikzstyle{input} = [coordinate]
\tikzstyle{output} = [coordinate]
\tikzstyle{sum} = [draw, circle]

\tikzstyle{startstop} = [rectangle, rounded corners, text centered, draw=black, fill=blue!10]
\tikzstyle{process} = [rectangle, minimum width=2cm, minimum height=1cm, text centered, draw=black, fill=orange!10]
\tikzstyle{decision} = [diamond, aspect=3, minimum width=3cm, minimum height=1cm, text centered, draw=black, fill=red!10]
\tikzstyle{compute} = [rectangle, minimum width=2cm, minimum height=1cm, text centered, draw=black, fill=green!10]
\tikzstyle{estimate} = [rectangle, rounded corners, minimum width=2cm, minimum height=1cm, text centered, draw=black, fill=yellow!10]
\tikzstyle{arrow} = [thick,->,>=stealth]

\begin{document}

\chapter{How to implement the Bayes' formula in the age of ML?}\label{chap1}

\author[1]{Amirhossein Taghvaei}%
\author[2]{Prashant G. Mehta}%

\address[1]{\orgname{University of Washington}, \orgdiv{Department of Aeronautics and Astronautics}, \orgaddress{Seattle, WA, USA}}
\address[2]{\orgname{University of Illinois at Urbana-Champaign}, \orgdiv{Department of Mechanical Engineering}, \orgaddress{Urbana, IL, USA}}

\articletag{}

\maketitle

\begin{glossary}[Keywords]
\begin{tabular}{@{}lp{34pc}@{}}
        Bayes formula \\
	Bayesian methods \\
        Control of probability distributions \\
        Estimation \\
        Feedback particle filter \\
	Kalman filter \\
        Nonlinear filtering \\
        Optimal transportation theory \\
	Posterior sampling \\
        Simulation-based algorithms \\
        Stochastic systems
\end{tabular}
\end{glossary}

\begin{glossary}[Key points/Objectives]
  \begin{itemize}
  \item A historical survey of the traditional filtering algorithms such as the Kalman filter and the particle filter.
\item A discussion of the fundamental limitations of these algorithms in non-Gaussian and high-dimensional applications.
  \item To overcome some of these limitations, a new optimization formulation of the Bayes' law is described, based on the optimal transportation (OT) theory.
\item The formulation is related to the feedback particle filter in continuous-time settings of the nonlinear filtering problem. 
\item The formulation combined with machine learning tools can lead to scalable algorithms in non-Gaussian and high-dimensional applications.  
  \end{itemize}
%

\end{glossary}

\begin{glossary}[Nomenclature]
\begin{tabular}{@{}lp{34pc}@{}}
	KF & Kalman Filter \\
	EnKF & Ensemble Kalman Filter \\
	PF & Particle Filter \\
        FPF & Feedback Particle Filter \\
	SIR & Sequential Importance (re)-sampling \\
	OT & Optimal Transportation \\
	SDE & Stochastic Differential Equation \\
	PDE & Partial Differential Equation
\end{tabular}
\end{glossary}

\begin{abstract}[Abstract]
This chapter contains a self-contained introduction to the significance of Bayes' formula in the context of nonlinear filtering problems.  Both discrete-time and continuous-time settings of the problem are considered in a unified manner. In control theory, the focus on optimization-based solution approaches is stressed together with a discussion of historical developments in this area (from 1960s onwards).  The heart of this chapter contains a presentation of a novel optimal transportation formulation for the Bayes formula (developed recently by the first author) and its relationship to some of the prior joint work (feedback particle filter) from the authors.  The presentation highlights how optimal transportation theory is leveraged to overcome some of the numerical challenges of implementing Bayes' law by enabling the use of machine learning (ML) tools. 
\end{abstract}

\newcommand{\tp}{^{\mathrm{T}}}

\section{Introduction}\label{sec:intro}

Recently we attended a stimulating week-long workshop in the scenic setting of Banff, Alberta.  While the workshop was not on the topic of nonlinear filtering, one of the attendees -- a giant in the field -- asked the following question:
\begin{quote}
  \centering
``How has nonlinear filtering changed compared to 1960s?'' 
\end{quote}
It is impossible to answer this question in a single article or even a single book. Our modest goal is to discuss some of the difficulties in this subject, describe some historical developments (in 1960s and 1990s) as well as some more recent advances, and offer some future perspectives for students and researchers.

In its simplest form, a nonlinear filter is a recursive application of the Bayes' formula. For jointly distributed random variables $(X,Y)$, the formula is given by
\begin{equation*}
	\text{(Bayes' formula)}\qquad  P_{X|Y}(x|y) = \frac{P_{Y|X}(y|x) P_{X}(x)}{P_Y(y)}
\end{equation*} 
In applications, $Y$ is observed (given by $Y=y$) and the goal is to compute the conditional probability $P_{X|Y}(\cdot|y)$ using the Bayes' formula.  While easily stated, it is notoriously hard to {\em implement} the Bayes formula.  This difficulty explains perhaps the pessimism implicit in the question heard at the Banff workshop.

\subsection{Implementing Bayes' formula in 1960s}

In the chapter~12 of the classical textbook~\cite{bryson1975applied} by Arthur Bryson and Larry Ho\footnote{Although the textbook is published in 1975, it is based on lecture notes for graduate courses offered at Harvard University during 1960s.}, the following model is introduced for the random variables $(X,Y)$:
\begin{equation}\label{eq:linear_model_intro}
\text{(model)}\qquad Y = H X + W
\end{equation}
where $H$ is a deterministic (known) matrix and $X$ and $W$ are independent Gaussian random vectors.  Suppose $X$ has mean $m$ and variance $\Sigma$ and $W$ is mean zero with variance $R$.  Random vector $W$ has the meaning of independent noise which is the reason for the mean zero assumption.  Suppose also that the variances $\Sigma$ and $R$ are positive-definite.  The problem is to compute the ``best estimate'' of $X$ given an observation $Y=y$.  The Bayes' formula provides an answer.  Instead of applying the Bayes' formula directly, Bryson and Ho introduce an optimization problem\footnote{The optimization problem is derived from the Bayes' formula via a maximum log-likelihood procedure.} as follows: 
\[
\text{(optimization problem)}\qquad J(x) = (x-m)\tp \Sigma^{-1} (x-m) + (y-Hx)\tp R^{-1}  (y-Hx) 
\]
Its minimizer is easily computed, e.g., by using either the completion-of-squares procedure or setting the derivative to zero (approach followed in~\cite[Ch.~12]{bryson1975applied}) to obtain
\[
\text{(Kalman filter update)}\qquad x^{\text{MIN}} = m + K (y- H m)
\]
where $K = \Sigma H\tp (H\Sigma H^\top + R )^{-1}$ is referred to as the Kalman gain.  This formula for the minimizer is at the heart of Kalman filter equation for linear Gaussian stochastic processes.  The formula is referred to as the (Bayesian) update formula for the Kalman filter.

In control theory, it is a preferred approach to implement the Bayes formula by posing and solving an optimization problem. For the model~\eqref{eq:linear_model_intro}, an alternate approach is based on minimizing the mean-squared-error (m.s.e.).  This m.s.e. minimization approach is followed in the seminal paper~\citep{kalman1960new}.  A side-by-side comparison of the two optimization cost functions and the respective solutions can be found in~\citep[Table 3.1]{kailath2000linear}.

\subsection{Aims of this article}

In this expository article, we have the following goals:
\begin{enumerate}
\item Define the problem of implementing the Bayes' formula in general settings (beyond the simple example above).
\item Describe the historical approaches, namely, importance sampling and ensemble Kalman filter, and discuss their limitations. 
\item Describe a novel formulation for the Bayes' formula as an optimization problem, specifically, as an optimal transport problem.
\item Relate this optimization problem to some of the more recent algorithms for nonlinear filtering, specifically, the feedback particle filter.
\item Offer some perspectives related to machine learning (ML).  Specifically,
\begin{enumerate}
\item How advances in ML (neural networks) informs the algorithmic development, and
\item How the proposed algorithms can be used for ML applications related to posterior sampling?
\end{enumerate}
\end{enumerate}

Before presenting the paper outline, we comment on centrality of optimization both in 1960s and in 2020s. While these problems used to be least squares type in 1960s~\citep{swerling1971modern}, optimal transportation provides the appropriate methodological framework in the current era. Several reasons for this are discussed as part of this paper.      

\subsection{Paper outline}

The outline of the remainder of this paper is as follows.  The nonlinear filtering problem and its particle filter solution is introduced in \Sec{sec:NF_PF}. The solution is helpful to introduce the problem of implementing the Bayes formula in \Sec{sec:IS_ENkF}.  Along with the problem, a historical survey of main solution approaches and their limitations is also described.  This is followed by~\Sec{sec:OT_Bayes} where the optimal transport (OT) formulation of the Bayes' formula is described.  This section also contains its relationship to the feedback particle filter (FPF).  In~\Sec{sec:OT_PF}, the OT formulation is used to develop the optimal transport filter.  Some numerical results are presented in~\Sec{sec:numerics} and conclusions in~\Sec{sec:conc}.

\section{Nonlinear filtering and Bayes' formula}\label{sec:NF_PF}

A nonlinear filtering problem involves two stochastic processes:
\begin{enumerate}
\item A \textit{state process}, denoted by $X:=\{X_t \in \Re^n;\, t \in \mathbb T\}$, that represents the hidden state of a dynamical system. 
\item An \textit{observation process}, denoted by $Y:=\{Y_t \in \Re^m;\, t \in \mathbb T\}$, that represents the observed sensor output or measurement data.
 \end{enumerate}
  The problem is formulated either in discrete-time or in continuous-time setting.  In discrete-time, $\mathbb T$ is the set of non-negative integers $\{0,1,2,\ldots\}$.  In continuous-time, $\mathbb T$ is set of non-negative real numbers $[0,\infty)$.  

\subsubparagraph{Discrete-time model:}  The state and observation processes evolve according to a probabilistic relationship.  For $t=1,2,3,\ldots$,
\begin{subequations}
	\begin{align}\label{eq:model-dyn-discrete}
	\text{(state)}	\qquad X_{t} &\sim  a(\cdot \mid X_{t-1}),\quad X_0 \sim \pi_0,\\\label{eq:model-obs-discrete}
	\text{(observation)}	\qquad	Y_t &\sim h(\cdot \mid X_t),
	\end{align}
where 
	$a(\cdot \mid \cdot)$ is the transition probability kernel of the next state $X_{t+1}$ given the current state $X_t$ (i.e. $\mathbb P(X_{t+1}\in A|X_t=x)=a(A|x)$ for any measurable set $A\subset \Re^n$), $h(\cdot \mid \cdot)$ is the transition probability kernels  of $Y_t$ given $X_t$, and $\pi_0$ is the probability measure for the random initial state $X_0$ (i.e, $\mathbb P(X_0\in A)=\pi_0(A)$). The  observation kernel  $h(\cdot \mid \cdot)$ is assumed to admit density $\ell(\cdot\mid \cdot)$, i.e.  $h(B\mid x)=\int_B \ell(y\mid x)\ud y$ for any measurable set $B\subset \Re^m$. The function $\ell(y\mid x)$ represents the likelihood of observing $Y_t=y$ given $X_t=x$.  

In many engineering applications of interest, e.g., target state estimation, it is natural to model $(X,Y)$ as a continuous-time process. In recent years, such continuous-time models have become popular in ML applications related to the diffusion models~\citep{ho2020denoising}.    

\subsubparagraph{Continuous-time model:} The state and observation processes are modeled with stochastic differential equations (SDEs)
	\begin{align}\label{eq:model-dyn-cont}
\text{(state)}	\qquad 	\ud X_t& = a(X_t) \ud t + b(X_t)\ud B_t,\quad X_0 \sim \pi_0,\\\label{eq:model-obs-cont}
	\text{(observation)}	\qquad\ud Y_t &=  h(X_t) \ud t + \ud W_t,
\end{align}
        where $a:\mathbb R^n \to \mathbb R^n$, $b:\mathbb R^n \to \mathbb R^{n\times n}$, and $h:\mathbb R^n \to \mathbb R$ are $C^1$ smooth globally Lipschitz functions, and $B:=\{B_t:t\geq 0\}$ and $W:=\{W_t:t\geq 0\}$ are standard Wiener processes.  A standard assumption is that $X_0$, $B$ and $W$ are mutually independent.  An important special case is the linear Gaussian model where $a(\cdot)$ and $h(\cdot)$ are linear functions, $b(\cdot)$ is a constant matrix, and $\pi_0$ is a Gaussian measure. For simplicity, we have restricted the continuous-time model to the case where the observation is one-dimensional, while the results generalize to the vector-valued observation setting.  

        \bigskip
        
 In either setting, continuous or discrete-time, the nonlinear filtering objective is to compute the conditional probability of the hidden state $X_t$, given the history of observations ($\sigma$-algebra) $\mathcal Y_t:=\sigma\{Y_s;\,s \in \mathbb T \cap [0,t]\}$. The conditional probability measure at time is denoted by $\pi_t$, defined so that, for any measurable set $A\subset \Re^n$, 
\begin{equation}\label{eq:posterior}
	\pi_t(A) := \mathbb P(X_t \in A \mid \mathcal Y_t),\quad \text{for}\quad t \in \mathbb T
\end{equation}
\end{subequations}
The conditional probability $\pi_t$ is also referred to as the {\it nonlinear filter}, the {\it posterior}, or the {\it belief state}.  For the linear Gaussian model, $\pi_t$ is a Gaussian measure whose mean and variance evolve according to the equations of the Kalman filter.  For the continuous-time setting, we assume that $\pi_t$ admits a (Lebesgue) density which is denoted as $p_t$.  That is $\pi_t(A) = \int_A p_t(x) \ud x$. 

\subsection{Nonlinear filter}

The nonlinear filter~\eqref{eq:posterior} admits a recursive structure which is  useful in the design of filtering algorithms. To present this recursive structure, we introduce the following operators in the discrete-time:
\begin{subequations}
	\begin{align}\label{eq:propagation}
		\text{(propagation)}\qquad \pi \mapsto \mathcal A \pi(A) &:= \int_{\mathbb R^n} a(A \mid x) \ud \pi(x),\\\label{eq:conditioning}
		\text{(Bayes update)}\qquad \pi \mapsto \mathcal B_y (\pi)(A) \!&:=\! \frac{\int_A\ell(y \mid x)\ud \pi(x)}{\int_{\mathbb R^n} \ell(y \mid x) \ud \pi(x)},\quad  \text{for all measurable subsets } A\subset \mathbb R^n.
	\end{align}
        The linear operator $\mathcal A$ is the \textit{forward operator}, referred to as such because it pushes forward (propagates) the measure using the dynamic model~\eqref{eq:model-dyn-discrete}. The nonlinear operator $\mathcal B_y$ is referred to as the \textit{Bayes update operator} because it implements the Bayes' formula.  In terms of these operators, $\pi_t$ follows the sequential update law \cite{cappe2009inference}
	\begin{align}\label{eq:exact-posterior}
		\pi_{t \mid t-1} &= \mathcal{A} \pi_{t-1},\quad 
		\pi_{t} =\mathcal{B}_{Y_t} (\pi_{t \mid t-1}),
	\end{align} 
	where the notation $\pi_{t \mid t-1}(\cdot):=\mathbb P(X_t \in \cdot \mid \mathcal Y_{t-1})$ is used to denote the conditional measure of $X_t$ before applying the Bayes update to account for the observation $Y_t$ made at time $t$. 
	We introduce the notation for the transition operator \[\mathcal T_{t,s}:=\mathcal{B}_{Y_t} \circ \mathcal A\circ \mathcal{B}_{Y_{t-1}}\circ  \mathcal A \circ \ldots  \circ \mathcal{B}_{Y_{s+1}} \circ \mathcal A\] to express $\pi_t = \mathcal T_{t,s}(\pi_s)$ for $t>s\geq 0$, and in particular, $\pi_t = \mathcal T_{t,0}(\pi_0)$.
\end{subequations}

\subsubparagraph{Continuous-time nonlinear filter:} This filter also admits a recursive structure. For any smooth bounded (test) function $f$, the continuous-time filter equation is given by
\[
\ud \pi_t(f) = \underbrace{\pi_t({\cal D}f) \ud t}_{\text{(propagation)}} + \underbrace{\pi_t((h-\pi_t(h)f) (\ud Z_t - \pi_t(h)\ud t)}_{\text{(Bayes update)}}, \qquad t\geq 0
\]
where the filter is initialized at time $t=0$ using $\pi_0$ and ${\cal D}$ is the generator of the SDE~\eqref{eq:model-dyn-cont}.  On the right-hand side, the first term is the propagation and the second term is the Bayes update.

\bigskip

Although the nonlinear filter equation is well established (since 1960s), the continuous-time filter may be unfamiliar to readers not steeped in the theory of nonlinear filtering.  The important take-away is the recursive structure of the filter involving propagation and the Bayesian update steps.  The two-step structure is common to both continuous-time and discrete-time settings of the problem. In fact, an elementary derivation of the continuous-time nonlinear filter is based on taking an appropriate limit of the two discrete-time operators $\mathcal A$ and $\mathcal {B}_{\Delta Y_t}$~\cite[Sec.~6.8]{jazwinski2007stochastic}.

\subsection{Particle filter}

Any numerical approximation of the nonlinear filter requires a finite-dimensional representation of the conditional measure $\pi_t$. Exact finite-dimensional representations are known only in a few special cases, the most prominent of which is the linear Gaussian case. Even in the linear Gaussian case, an exact implementation becomes computationally intractable in very high dimensions, e.g., in applications in geo-sciences and weather prediction~\citep{evensen2006,houtekamer01}.

All of this has motivated a Monte-Carlo or particle-based approaches where the posterior is approximated by the empirical distribution of an ensemble of particles $\{X^1_t,\ldots,X^N_t\}$ as follows: 
\begin{align*}
	\pi_t \approx \frac{1}{N}\sum_{i=1}^N \delta_{X^i_t}
\end{align*}
where $\delta_x$ is the dirac delta distribution at $x$ (see Figure~\ref{fig:empitical-approx}).  In terms of the particles, the implementation of the propagation step is straightforward: simply use the kernel $a(\cdot\mid \cdot)$ to sample in discrete-time or the SDE~\eqref{eq:model-dyn-cont} in continuous-time.  The main difficulty is the implementation of the Bayes' update step. This is the subject of the following section.

\begin{figure}[t]
	\centering 
	\includegraphics[width=0.35\hsize]{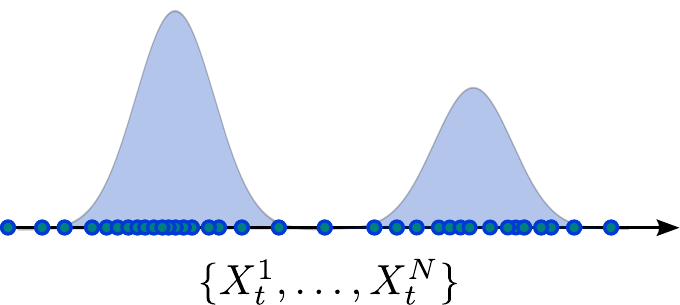}
	\caption{This exposition is concerned with particle filter algorithms that seek to approximate the posterior with the empirical distribution of an ensemble of particles.}
	\label{fig:empitical-approx}
\end{figure}

\begin{remark}
  Although their paper is often overlooked, Handschin and Mayne wrote a prescient paper, in 1960s, introducing the particle filter for the problem of nonlinear filtering~\citep{handschin1969monte}.  In part because of the emergence and wide-spread adoption of computers at the time, the field received an impetus from two highly cited papers published in early 1990s: 
  \begin{enumerate}
  \item~\cite{gordon1993novel} is widely regarded as introducing the modern particle filter~\citep{cappe2007overview}.
  \item~\cite{evensen1994sequential} introduced the ensemble Kalman filter.  This paper spawned an area known as {\em Data assimilation} which aims to develop and apply nonlinear filtering algorithms for high-dimensional applications such as weather prediction.  
  \end{enumerate}
\end{remark}

\section{Fundamental challenges in implementing the Bayes update} \label{sec:IS_ENkF}
This section is concerned with algorithms for implementing the Bayes operator ${\cal B}_y$ for the nonlinear filter.  Specifically, 
let $X$ be a random variable with probability law $P_X$. And let $Y \sim h(\cdot\mid X)$ be the observation random variable. Let $P_{X,Y}$ denote the joint probability law of $(X,Y)$ and let $P_{X|Y}(\cdot|y)$ be the conditional probability law of $X$ given the observation $Y=y$.   
In a particle-based setting, we have $N$ samples drawn i.i.d. from the prior $P_X$.  Our objective is to obtain $N$ samples from the conditional $P_{X|Y}(\cdot|y)$, where the value $y$ is allowed to be arbitrary.
\subsubparagraph{Bayes sampling problem:}
\begin{align*}
	\text{given:}&\qquad \{X^1,\ldots,X^N\} \sim P_X,
	\\
	\text{generate:}&\qquad \{X^1_{\mid y},\ldots,X^N_{\mid y}\} \sim P_{X|Y}(\cdot|y)\quad \text{for a given $y$.}
\end{align*}
It is important to note that the explicit form of $P_X$ is not assumed.  One {\em only} assumes $N$ particles sampled from $P_X$.  
Regarding our knowledge of the observation model, we consider two settings: 
\begin{enumerate}
	\item {\bf Analytical:} The likelihood function  is known in an explicit form. That is one can evaluate $\ell(y\mid x)$ for all $x\in\Re^n$ and $y\in\Re^m$.
	\item {\bf Simulation-based:} The likelihood function  is not available in an explicit form. Instead, one can generate  samples $Y \sim h(\cdot \mid x)$ for any given $x\in \Re^n$. 
\end{enumerate}

The algorithms that are presented below are either based on an analytical or simulation-based knowledge of the observation model.  We make the distinction clear when describing the algorithm.

For the purposes of error analysis, the following metric is introduced.  Denote by $\mathcal G$ the space of functions that are bounded in absolute value by one and that are Lipschitz with {a Lipschitz constant} smaller than one:
\[
\mathcal G:=\{g:\Re^n \to \Re:\; \;|g(x)|\leq 1 \;\; \text{and} \; \; |g(x)-g(x')|\leq |x-x'|,\quad \forall x,x'\in\Re^n\}
\]

\begin{definition}
Let  $\mu,\nu$ be possibly random probability measures. The {\it dual bounded-Lipschitz metric} is given by
\[
d(\mu,\nu) := \sup_{g \in \mathcal G} \sqrt{\mathbb{E}\left|\mu(g) - \nu(g)\right|^2}
\]
\end{definition}

\subsection{Bayes update using importance sampling and the curse of dimensionality}

A vanilla importance sampling and resampling (SIR) particle filter (PF) carries out this task by first forming a weighted empirical distribution and then resampling from the weighted distribution~\cite{gordon1993novel,arulampalam2002tutorial,gordon2004,doucet2001sequential,doucet09}.  The two steps are as follows:
	\begin{enumerate}
	\item \textbf{Step 1.} The importance weights are computed
                  \[
\text{(importance weights)}\qquad w^i = \frac{\ell(y|X^i)}{\sum_{j=1}^N \ell y|X^j)},\qquad i=1,2,\hdots,N
\] This step is
                  called importance sampling. The likelihood function is explicitly used in this step. 
		\item \textbf{Step 2.} Next $N$ particles are
                  independently sampled from the weighted
                  distribution:
                  \[
\text{(re-sampling)}\qquad	X_{\mid y}^i \sim P^{\text{SIR}}_{X|Y}(\cdot|y):=\sum_{j=1}^N w^j \delta_{X^j}
        \]
        by sampling from a multinomial
                  distribution with parameter vector
                  $(N,\{w_i\}_{i=1}^N)$. This step is called
                  resampling.
	\end{enumerate}
        
It is shown that SIR PF is consistent, and the approximation error decreases with the rate $O(\frac{1}{\sqrt{N}})$~\cite{del2001stability}.
However, the PF is known to perform poorly because of the problem of weight degeneracy, whereby all but few particles have negligible (nearly zero) weights.  The problem is known to become worse in high dimensions and is referred to as the curse of dimensionality (CoD) for PF~\citep{doucet09,gordon2004,bickel2008sharp,bengtsson08,beskos2014error,rebeschini2015can}. A quantitative result on the CoD is given in the following Proposition. 
\begin{proposition}
	Consider the SIR procedure for approximating the conditional distribution. Suppose $\overline X$ is an independent copy of $X$.  Assume $\mathbb E[\frac{\ell(Y|\overline X)^2}{\mathbb E[\ell(Y|\overline X)|Y]^2}]<\infty$. Then, 
	\begin{align}\label{eq:SIR-bound}
		\liminf_{N \to \infty}\sqrt{N} d(P_{X|Y}^{\text{SIR}}(\cdot|Y),P_{X|Y}(\cdot|Y)) &\geq  \sup_{g \in \mathcal G}\sqrt{V(g)}
	\end{align}
	where $V(g):=\mathbb E[\frac{\ell(Y|\overline X)^2}{\mathbb E[\ell(Y|\overline X)|Y]^2}(g(\overline X)-\mathbb E[g(X)|Y])^2]$. In particular, for the special case where $X=[X_1,X_2,\hdots,X_n]$ and $Y=[Y_1,Y_2,\hdots,Y_n]$ are $n$-dimensional random vectors with independent and identically distributed components (i.e., $(X_j,Y_j)$ and $(X_i,Y_i)$ are mutually independent for $i\neq j$), then there exists  $N_0$ large enough such that  for all $N>N_0$: 
	\[d(P_{X|Y}^{\text{SIR}}(\cdot|Y),P_{X|Y}(\cdot|Y)) \geq  \frac{C\gamma^n}{\sqrt{N}}\]
	where $C$  and $\gamma>1$ are constants, independent of the dimension $n$. 
\end{proposition}
\begin{theorem*}[Proof] 
	The proof is based on the application of central limit theorem  for importance sampling~\cite[Thm. 9.1.8]{cappe2009inference} and appears in~\cite[Appendix B.4.1]{al2023high}. 
\end{theorem*}

\newcommand{\NN}{\mathcal{N}}
\newcommand{\KN}{{\sf K}^{(N)}}
\newcommand{\K}{{\sf K}}
\newcommand{\mN}{m^{(N)}}
\newcommand{\SigN}{\Sigma^{(N)}}

\subsection{Bayesian update for the ensemble Kalman filter}

In the light of the negative result for PF, let us re-visit the linear model~\eqref{eq:linear_model_intro} introduced in \Sec{sec:intro} where the samples $X^i \stackrel{\text{i.i.d.}}{\sim} N(m,\Sigma)$  for $i=1,2\hdots,N$.  Consistent with our assumption, explicit knowledge of $m$ and $\Sigma$ is not assumed.  Of course, these may be estimated empirically using the particles.  These estimates are denoted by
\[
m^{(N)}:= \frac{1}{N} \sum_{i=1}^N X^i,\qquad \Sigma^{(N)}:= \frac{1}{N-1} \sum_{i=1}^N (X^i - m^{(N)})(X^i - m^{(N)})^{\rm T}
\]
Our goal is to sample $\{X_{\mid y}^i:i=1,2,\hdots,N\}$ from the (Gaussian) posterior.

In order to explain the EnKF update, it is useful to introduce some notation.  This notation will be used consistently in the remainder of the paper in other more general settings of the problem.  Let $\overline{X}\sim P_{X}$ be an independent copy of $X$.  That is, $\overline{X}$ has the same Gaussian distribution $N(m_0,\Sigma_0)$ as the prior.  Using the model~\eqref{eq:linear_model_intro}, consider
\[
\overline{Y} = H \overline{X} + \overline{W}
\]
where $\overline{W} \sim N(0,R)$ is an independent copy of $W$.  Then $(\overline{X},\overline{Y})\sim P_{X,Y}$.  The EnKF update is based on defining the following affine map:
\begin{subequations}
\begin{equation}\label{eq:ENKF_map}
T(\overline{X},y) := \overline{X} + K (y - \overline{Y})
\end{equation}
where $K= \Sigma H\tp(H\Sigma H\tp + R)^{-1}$ is the Kalman gain (same formula as the one described in \Sec{sec:intro}). The reason for defining the map appears in the following proposition.

\begin{proposition}\label{Prop:EnKF}
Set $\overline{X}_{\mid_y}  =  T(\overline{X},y)$.  Then $\overline{X}_{\mid_y} \sim P_{X|Y}(\cdot|y)$.  
\end{proposition}

\begin{proof}
Because the affine nature of the map, $\overline{X}_{\mid_y}$ is Gaussian.  The proof is completed by computing the mean and variance and showing these to be identical to the conditional mean and covariance. These calculations can be found in the Appendix. 
\end{proof}

Using particles, the EnKF update is expressed as 

\begin{equation}\label{eq:Xi-EnKF-discrete-particles}
\text{(EnKF Bayes update)}\qquad X_{\mid_y}^i = X^i + K^{(N)}( y - Y^i),\quad i=1,2,\hdots,N
\end{equation}
\end{subequations}
where $Y^i=HX^i+W^i$, $W^i \stackrel{\text{i.i.d.}}{\sim} N(0,R)$ are independent copies of the
observation noise, and $K^{(N)} = \Sigma^{(N)} H\tp( H\Sigma^{(N)}H\tp  +  R)^{-1}$ is the empirical approximation of the Kalman gain.  The empirical distribution is denoted by,
\begin{equation*}
	P^{\text{EnKF}}_{X|Y=y}(\cdot|y):= \frac{1}{N}\sum_{i=1}^N \delta_{X^i_1}
\end{equation*}
Note that the EnKF algorithm only requires a simulation-based assumption  about the observation model. 
Error analysis for the same is given in the following proposition. 
\begin{proposition}\label{prop:EnKF-error}
	Suppose $P_{X,Y}$ is Gaussian. Then 
	\begin{align}\label{eq:EnKF-bound}
		d(P_{X|Y}^{\text{EnKF}}(\cdot|Y),P_{X|Y}(\cdot|Y)) &\leq \frac{1}{\sqrt{N}} + 2 \mathbb E[\|Y\|^4]^{\frac{1}{4}}\mathbb E[\|(K^{(N)} - K)\|^4]^{\frac{1}{4}}.
	\end{align}
\end{proposition}
\begin{proof}
See Appendix. 
\end{proof}

The first term in the error bound~\eqref{eq:EnKF-bound} is due to the Monte-Carlo sampling error and the second term is due to empirical approximation of the gain matrix, which is expected to scale polynomially with the problem dimension. As a result, EnKF algorithm does not suffer from the curse of dimensionality. This is also verified in the continuous-time setting~\cite{surace2019avoid,taghvaei2020optimal}. However, the EnKF algorithm  gives an asymptotically exact approximation of the posterior only in the Gaussian setting.  This limitation motivated generalization of the EnKF algorithm to the non-Gaussian setting. In the next section, we present such a generalization based on optimal transportation theory.

Apart from~\eqref{eq:Xi-EnKF-discrete-particles}, there are other forms of the EnKF update.  One particular update -- that has been crucial in successful application of EnKF in geosciences -- is the
ensemble square-root Kalman filter (EnSRKF)~\citep{whitaker2002ensemble} and~\citep[Sec. 7.1]{reich2015probabilistic}.  

Although presented in the Gaussian setting, the EnKF update~\eqref{eq:Xi-EnKF-discrete-particles} may be also implemented for a non-Gaussian prior and nonlinear observation model, where the gain $K=\text{Cov}(X,Y)\text{Cov}(Y)^{-1}$ is approximated in terms of particles $(X^i,Y^i)_{i=1}^N$.

\subsection{Bayes update in feedback particle filter (FPF)}

The feedback particle filter (FPF) algorithm~\citep{Tao_TAC,yang2016} is designed to numerically approximate the solution to the continuous-time nonlinear filtering problem where the state and observations are modeled according to~\eqref{eq:model-dyn-cont}-\eqref{eq:model-obs-cont}. To illustrate the Bayesian update step, assume the state is static, i.e. $a(x)\equiv 0$ and $b(x)\equiv 0$, and $\pi_0$ admits a $C^2$ density denoted as $p_0$.  In this case, the FPF algorithm proceeds by simulating a controlled stochastic process
\begin{subequations}
\begin{equation}\label{eq:Xbar}
	\ud \overline{X}_t = \k_t(\overline{X}_t) \ud Y_t + u_t(\overline{X}_t) \ud t,\qquad \overline{X}_0\sim \pi_0
\end{equation}
where the vector-fields $\k_t$ and $u_t$ are designed  such that the probability density $\bar{p}_t$ of $\overline{X}_t$ coincides with the
posterior density $p_t$ for all $t\geq 0$. The derivation of the vector-fields $\k_t$ and $u_t$ follows by matching the evolution equations for two densities. 
Although the matching procedure does not lead to a unique specification of the vector-fields, the choice in FPF is as follows:
	\begin{align}\label{eq:K-u}
	\k_t:=\nabla \phi_t,\quad \text{and}\quad 
	u_t := -\frac{h+\hat{h}_t}{2} \nabla \phi_t + \frac{1}{2}\nabla^2 \phi_t \nabla \phi_t
	\end{align}
	 where $\phi_t$ solves the Poisson equation 
	\begin{equation}
	- \frac{1}{p_t(x)} \nabla \cdot (p_t(x)\nabla \phi_t(x)) = h(x) - \overline{h}_t,\quad \forall x \in \Re^n,
	\end{equation} 
        and $\overline{h}_t = \int_{\Re^n} h(x) p_t(x) \ud x$.
\end{subequations}
Combining these the FPF Bayesian update formula is given by
\[
\text{(FPF Bayes update)}\qquad \ud \overline{X}_t = \nabla \phi_t (\overline{X}_t) \circ \left(\ud Y_t - \frac{h(\overline{X}_t) + \overline{h}_t}{2} \ud t \right), \qquad \overline{X}_0\sim \pi_0
\]
where $\circ$ means that the SDE is expressed in its Stratonovich form.  A particle form of the FPF update is obtained by approximating the terms empirically, e.g., $\overline{h}_t \approx \overline{h}_t^{(N)} = N^{-1} \sum_{i=1}^N h(X_t^i)$.  The main difficulty lies in solving the Poisson equation and the particle-based algorithms for the same appear in~\citep{taghvaei2020diffusion} (see also~\cite[Sec.~4]{taghvaei2023survey} and~\cite[Sidebar on page~46]{taghvaei2021optimal}).

\subsection{Summary}

  For both EnKF and the FPF algorithms, the Bayes update formula is based on two steps:
  \begin{enumerate}
  \item Defining a random variable $\overline {X} \sim P_{X}$ as an independent copy of $X$.  In an FPF, this step involved defining $\overline {X}_0 \sim \pi_0$.
  \item Defining a map $T=T(\overline {X}, Y)$ where $Y$ is observed random variable.  In an FPF, the observations $\{Y_s:0\leq s\leq t\}$ are over a time-period $[0,t]$.  So, the map $T=T(\overline {X}_0 , \{Y_s:0\leq s\leq t\})$.  The FPF SDE is merely a convenient description to implement the map in continuous-time settings of the problem.  
  \end{enumerate}

  In the following section, the central result of this paper is described.  While the result is motivated by the EnKF and FPF update formulae, and recent (since 2010s) developments in the application of optimal transport for Bayesian inference (specifically, the ground-breaking works of~\citep{el2012bayesian,spantini2022coupling} and~\citep{reich2013nonparametric}), the form presented below appeared for the first time in a paper from the first author and collaborators in UW Seattle~\citep{taghvaei2022optimal,al2023optimal,al2023high}.

  \section{Bayes update with optimal transport maps (a.k.a. ``implementing Bayes update in 2020s'')}\label{sec:OT_Bayes}

We consider the Bayes sampling problem introduced in the preceding section with the simulation-based assumption about observation model.   The design problem we seek to solve is the following:  

\subsubparagraph{Problem:} Design a map $T$ such that 
\begin{equation}\label{eq:consistency-condition}
	T(\cdot,y)_{\#} P_X = P_{X|Y}(\cdot|y),\quad \forall y.
\end{equation}
where $\#$ denotes the push-forward operator.

Once the map is known then the Bayes sampling problem is solved simply as $X^i_{\mid y} = T(X^i,y)$ for $i=1,2,\ldots,N$. 
The following example is illustrative. 

\begin{example}[Noiseless observation]
	Suppose  $Y=h(X)$ where $h$ is an invertible map. Then, the conditional $P_{X|Y}(\cdot|y)= \delta_{h^{-1}(y)}(\cdot)$ is simply given by a transport map $T(x,y) = h^{-1}(y)$. 
\end{example}

We refer to~\eqref{eq:consistency-condition} as the consistency condition. There are many possible choices of maps $T$ that satisfy the consistency condition.  For this reason, it is natural to use the theory of optimal transportation~\cite{villani2003topics} to uniquely select a map $T$ that satisfies~\eqref{eq:consistency-condition}. Before describing the procedure, we introduce some notation and definitions as follows:
\begin{enumerate}
\item $P_X \otimes P_Y$ is used to denote the joint distribution obtained by multiplying the two marginals $P_X$ and $P_Y$.  Such a joint distribution is referred to as the independent coupling.
\item $\mathcal{M}(P_X \otimes P_Y)$ is the set of maps from $\Re^n \times \Re^m$ to $\Re^n$ that are $P_X \otimes P_Y$-measurable.
\item The cost function $c(x,x')=\frac{1}{2}|x-x'|^2$.
\item  A function $f(x)$ is said to be $c$-concave if $\frac{1}{2}|x|^2-f(x)$ is convex.
\item A function $f(x,y)$ is said to be $c$-concave$_x$  if $\frac{1}{2}|x|^2-f(x,y)$ is convex (i.e., the function is $c$-concave  in the $x$-argument). 
\end{enumerate}

\bigskip

The procedure is as follows.

\subsubparagraph{Step 1.}
Suppose $\overline X$ is an independent copy of $X$.  Replace the condition \eqref{eq:consistency-condition} with 
	\begin{equation}\label{eq:T-constraint-joint}
	(T(\overline X,Y),Y) \sim P_{X,Y}.
\end{equation}
        The justification for~\eqref{eq:T-constraint-joint} is as follows.  Condition~\eqref{eq:T-constraint-joint} implies $\mathbb E[f(T(\overline X,Y))g(Y)] =   \mathbb E[f(X)g(Y)]$ for all measurable and bounded functions $f:\Re^n \to \Re$ and $g:\Re^m \to \Re$, concluding that $T(\overline X,Y) \sim P_{X|Y}(\cdot|Y)$ by the definition of conditional expectation. A more rigorous justification appears in~\cite[Appendix B.1]{al2023high} and~\cite[Thm. 2.4] {kovachki2020conditional}.

\subsubparagraph{Step 2.}
	In order to select a unique map that satisfies the condition~\eqref{eq:T-constraint-joint}, we formulate the (conditional) Monge problem under quadratic cost:
	\begin{equation} \label{eq:Monge OT}
		\begin{aligned}
			\min_{T \in \mathcal{M}(P_X \otimes P_Y)}\, &\mathbb E\left[c(T(\overline X,Y),\overline X)\right],\quad \text{s.t.~ (\ref{eq:T-constraint-joint}) holds.} 
		\end{aligned}
	\end{equation} 
	The optimization~(\ref{eq:Monge OT}) is viewed as the Monge problem between the independent coupling $(\overline X,Y) \sim P_X \otimes P_Y$ and the joint distribution $(X,Y)\sim P_{X,Y}$ with transport maps that are constrained to be block-triangular $(x,y) \mapsto (T(x,y),y)$.

\subsubparagraph{Step 3.}
	Upon using the Kantorovich duality and the definition of $c$-concave function, 
	the Monge problem~(\ref{eq:Monge OT}) becomes 
	\begin{align}\label{eq:new_loss}
		\max_{ f \in c\text{-concave}_x}\,
		\min_{T \in \mathcal{M}(P_X \otimes P_Y)}\, J(f,T;P_{X,Y}),
	\end{align}
	where the objective function 
	\begin{equation*}
		J(f,T;P_{X,Y}):=\mathbb E_{(X,Y) \sim P_{X,Y}}[f(X,Y)] -  \mathbb E_{(\overline X,Y) \sim P_X \otimes P_Y} \left[f(T(\overline X,Y),Y)+  c(T(\overline X,Y),\overline X)\right].
	\end{equation*}

\medskip

A rigorous justification of the max-min formulation, in the standard OT setting, appears in \cite{al2023high}, where the following result is described concerning its solution (the result is based on \cite[Theorem~2.3]{carlier2016vector}).  
\begin{proposition} \label{prop:consistency}
	Assume $P_X$ is absolutely continuous with respect to the Lebesgue measure with a convex support set $\mathcal{X}$, $P_{X|Y}(\cdot|y)$ admits a density with respect to the Lebesgue measure 
	$\forall y$, and $\mathbb E[|X|^2]<\infty$. Then, there exists a unique pair $(\overline f,\overline T)$, modulo an additive constant for $\overline{f}$, that solves the optimization problem~(\ref{eq:new_loss}) and  the map $\overline T(\cdot,y)$ is the OT map from $P_X$ to $P_{X|Y}(\cdot|y)$ for a.e. $y$. 
\end{proposition}

The next proposition is concerned with the case where the max-min optimization problem~(\ref{eq:new_loss}) is not solved exactly.  For a pair $(f,T)$, the total optimality gap for the max-min problem is defined as
\begin{equation}\label{eq:opt-gaps}
\begin{aligned}
\epsilon(f,T;P_{X,Y}) := &J(f,T;P_{X,Y}) - \min_{S} J(f,S;P_{X,Y}) +\max_{g}\min_{S}J(g,S;P_{X,Y}) - \min_{S} J(f,S;P_{X,Y}) 
\end{aligned}
\end{equation}
We then have the following result (which is as an extension of \cite[Thm. 4.3]{rout2022generative} and \cite[Thm. 3.6]{makkuva2020optimal}).

\begin{proposition} \label{prop:map estimation error bound}
Consider the setting of Prop.~\ref{prop:consistency} with the optimal pair $(\overline f, \overline T)$. Let $(f,T)$ be a possibly non-optimal pair with an optimality gap $\epsilon(f, T;P_{X,Y})$.  Suppose $x \mapsto \frac{1}{2}|x|^2-f(x,y)$ is $\alpha$-strongly 
convex in $x$ for all $y$. Then,  
\begin{equation}\label{eq:OT-bound}
d(P_{X|Y}(\cdot|Y), T(\cdot,Y)_{\#} P_X) \leq \sqrt{\frac{4}{\alpha}\epsilon(f,T;P_{X,Y})}.
\end{equation}
\end{proposition}

\begin{remark}
	The OT upper-bound~\eqref{eq:OT-bound} depends on the optimality gap 
	$\epsilon(  f,  T)$ which, in principle, decomposes to a bias and variance term. The bias term corresponds to the representation power of the function classes $\mathcal{F}$ and $\mathcal T$, in comparison with the complexity of the problem. The variance term 
	corresponds to the statistical generalization errors due to  the empirical approximation of the objective function. The variance term is expected to grow as $O(\frac{1}{\sqrt{N}})$ with a proportionality constant that depends on the complexity of the function classes, but independent of the dimension.  In principle, the OT approach may also suffer from the COD under no additional assumptions on the problem. However, in comparison to SIR, it provides a more flexible design methodology that  can exploit problem specific structure and regularity.
\end{remark}
\subsection{Numerical approximation of the OT map}

The max-min problem structure is appropriate in this age of ML where neural network architectures for function approximation and optimization frameworks for the same are readily available.  The data for the same is obtained from sampling as follows:
\begin{align*}
  \text{(sample)} & \qquad (X^i,Y^i) \overset{\text{i.i.d}}{\sim} P_{X,Y} \quad i=1,2,\hdots, N\\
  \text{(shuffle)} & \qquad {\overline X}^i = X^{\sigma_i} \quad i=2,\hdots,n
\end{align*}
where $\sigma: \{1,\ldots,N\} \to \{1,\ldots,N\}$ is a random shuffling.  Because of shuffling $({\overline X}^i,Y^i)\sim P_X \otimes P_Y$. 

Using these samples, the objective function $J(f,T)$ is approximated empirically as
\begin{equation}\label{eq:emperical_loss}
	\begin{aligned}
		&J(f,T;\frac{1}{N}\sum_{i=1}^N \delta_{(X^i,Y^i)}):= \frac{1}{N} \sum_{i=1}^N \big[f(X^i,Y^i) + \frac{1}{2}| T( \overline X^i,Y^i)- X^i|^2 - f(T(\overline X^i,Y^i),Y^i)\big]
	\end{aligned}
\end{equation}

The function $f$ and the map $T$ are represented with a parametric class of functions, denoted by $\mathcal F$ and $\mathcal T$, respectively. Here we take these to be neural network classes with architectures 
that are summarized in Fig.~\ref{tikz:static_struc}; further details 
about these architectures can be found in 
the \cite[Appendix C]{al2023high}.  The max-min optimization problem becomes
\begin{equation}\label{eq:empirical-optimization}
	\max_{f\in \mathcal F}\,\min_{T \in \mathcal T}\, J(f,T;\frac{1}{N}\sum_{i=1}^N \delta_{(X^i,Y^i)}). 
\end{equation}

\begin{remark}\label{remark:TH}
	Note that our choice of $\mathcal F$ does not impose the constraint that $f(x,y)$ 
	is $c\text{-concave}_x$. We make this choice due to the practical limitations of imposing convexity constraints on 
	neural nets using, e.g., input-convex networks \cite{amos2016input,bunne2022supervised}. However, if the computed 
	$f$ happens to be $c\text{-concave}_x$ (which one can check a posteriori) then Prop.~\ref{prop:map estimation error bound} remains applicable. 
\end{remark}

\begin{remark}
	The proposed computational procedure may be extended to the Riemannian manifold setting by using the square of the geodesic distance as the cost function $c$ and modeling the map $T$ as exponential of a parameterized vector-field; see~\cite{grange2023computational}.
\end{remark}

\begin{figure}[t]
	\scriptsize
	\centering 
	\begin{tikzpicture}[node distance=1cm, auto, scale = 0.8]
		\node (input) [startstop,minimum width=1cm, minimum height=0.5cm,rotate=270] at (0,0) {$X$, $Y$ \par};
		
		
		\node (x) [input] at (-3,-0.45) {};
		\node (y) [input] at (-3,0.45) {};     
		
		\node (resnet) [process,minimum width=1cm, minimum height=1cm, text width=1cm] at (1.75,0) {ResNet Block \par};
		
		\node (f) [process,minimum width=1cm] at (4.75,0){$f(X,Y)$ \par };
		\draw[arrow] (x) -- node [above]{$X$}($(input.270) + (0,-0.45)$);
		\draw[arrow] (y) -- node [above]{$Y$}($(input.270) + (0,0.45)$);
		\draw[arrow] (input) -- (resnet);
		\draw[arrow] (resnet) -- (f);

		
	\end{tikzpicture}
	
	
	\begin{tikzpicture}[node distance=1cm, auto, scale=0.8]
		\node (input) [startstop,minimum width=1cm, minimum height=0.5cm,rotate=270] at (0,0) {$X$ , $Y$ \par};
		\node (enkf) [process,fill=red!10, rotate=270, text width=0.85cm, minimum width=1cm] at (-1.75,-0.5){EnKF Block\par};
		
		\node (x) [input] at (-3,-0.48) {};
		\node (y) [input] at (-3,0.45) {};        
		
		\node (resnet) [process,minimum width=1cm, minimum height=1cm, text width=1cm] at (1.75,0) {ResNet Block \par};
		
		\node (sum) [sum] at (3.25,0){$+$ \par };
		\node (T) [process,minimum width=1cm] at (4.75,0){$T(X,Y)$ \par };
		\draw[arrow] (x) -- node {$X$} (enkf);
		\draw[arrow] (y) -| node [above]{$Y$} (enkf);
		\draw[arrow] (y) -- ($(input.270) + (0,0.45)$);
		\draw[arrow] (enkf) -- ($(input.270) + (0,-0.48)$);
		\draw[arrow] (input) -- (resnet);
		\draw[arrow] (resnet) -- (sum);
		\draw[arrow] (sum) -- (T);
		
		\draw[arrow] (-0.75,-0.48) |- (3.25,-1) -| ($(sum.270)$);
		
	\end{tikzpicture}
	\caption{Neural net architectures for the function classes  $\mathcal F$ 
		and $\mathcal T$ within our proposed algorithm.}
	\label{tikz:static_struc}
\end{figure}
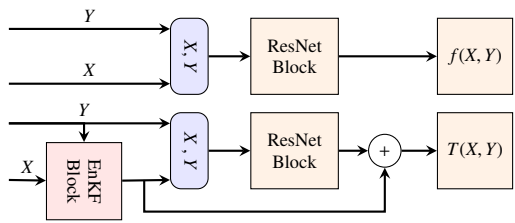

\subsection{Numerical demonstration}\label{sec:Static_Example}
We present a  numerical experiment to  demonstrate the performance of the OT approach in comparison with the EnKF and SIR algorithms. 
The OT algorithm 
consists of solving~\eqref{eq:empirical-optimization} with $f, T$ expressed using neural nets. The network weights are learned with a gradient ascent-descent procedure using the Adam optimization algorithm.  The details of the algorithm appear in~\cite{al2023high}. 

 Specifically, we consider the task of computing the conditional distribution of a Gaussian hidden random variable $X \sim N(0,I_n)$ given the observation
\begin{align} \label{example:squared}
	Y=\frac{1}{2}X\odot X + \lambda_w W, \qquad W \sim N(0,I_n)
\end{align}
where $\odot$ denotes the element-wise (i.e., Hadamard) product. This model is specifically selected to produce a bimodal posterior.  We only present the  $n=2$ case since the difference between OT and SIR was not significant when $n=1$. 

The first numerical results for this model are presented with a noise 
standard deviation of $\lambda_w=0.4$
in Fig.~\ref{fig:squared_not_SNR}. The top left panel shows the initial particles as samples from the Gaussian prior distribution. 
The bottom left panel shows the pushforward of samples from 
$P_X \otimes P_Y$ via 
the block triangular map $S(x,y)=(T(x,y),y)$, in comparison to samples from $P_{XY}$, verifying the consistency condition (\ref{eq:T-constraint-joint}) 
for the map. Then, we pick a particular value for the observation $Y=1$ (as shown by the dashed line) and present the histogram of (transported) particles in comparison with the exact conditional density. It is observed that both OT and SIR capture the bimodal posterior, while EnKF falls short since it always approximates the posterior with a Gaussian.

We repeat the procedure in
Fig.~\ref{fig:squared_SNR} but for a smaller noise standard deviation $\lambda_w=0.04$ which leads to a more degenerate posterior. Our results 
clearly demonstrate the  weight degeneracy of SIR even in this low-dimensional setting, as all particles collapse into a single mode, while the OT approach still captures the bimodal posterior. 

\begin{figure*}[t]
	\centering
	\begin{subfigure}{0.48\textwidth}
		\centering
		\includegraphics[width=0.95\textwidth, ,trim={70 20 90 40},clip]{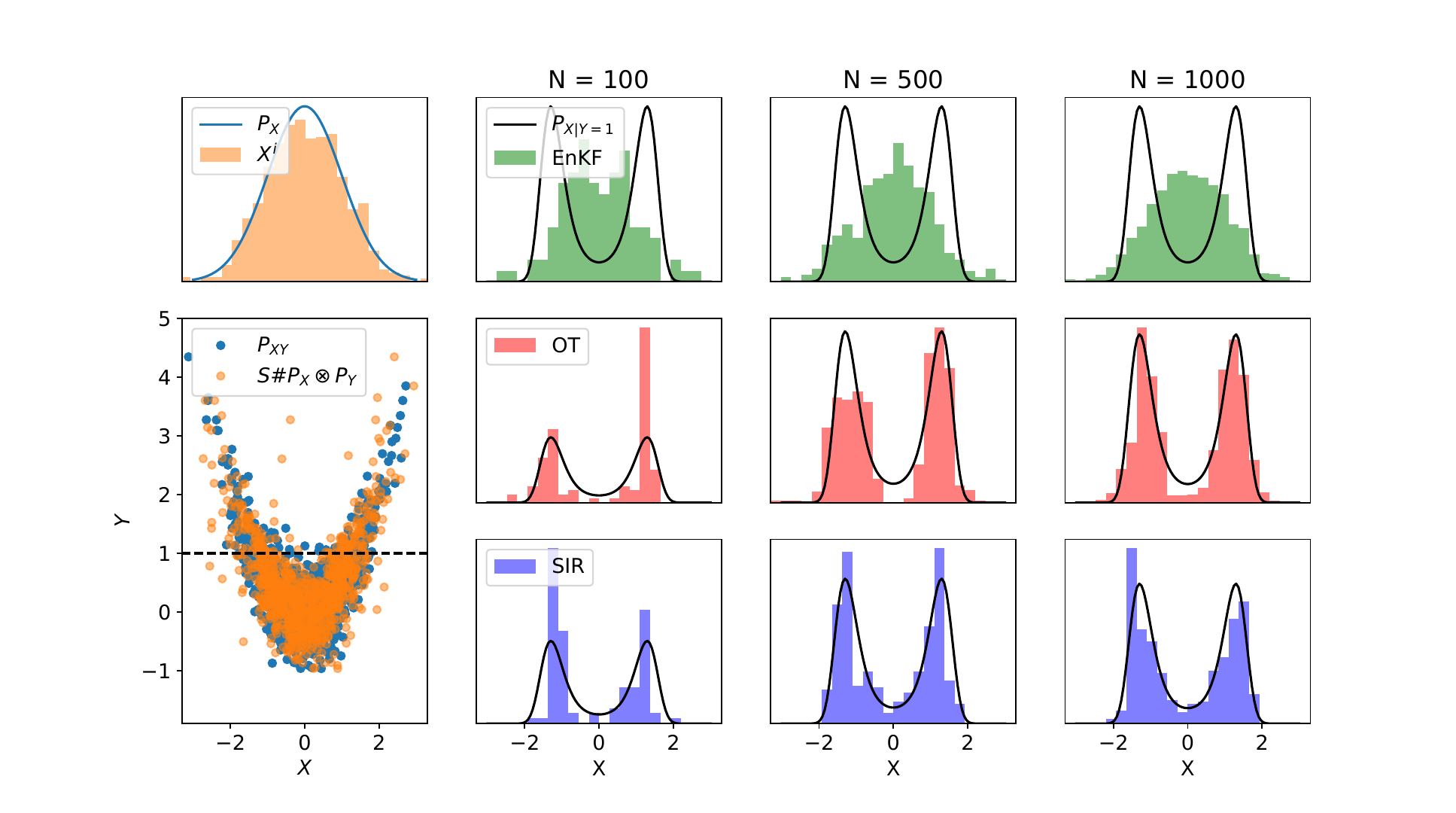}
		\caption{$\lambda_w=0.4$.}
		\label{fig:squared_not_SNR}
	\end{subfigure}
	\hfill
	\begin{subfigure}{0.5\textwidth}
		\centering
		\includegraphics[width=0.95\textwidth,trim={70 20 90 40},clip]{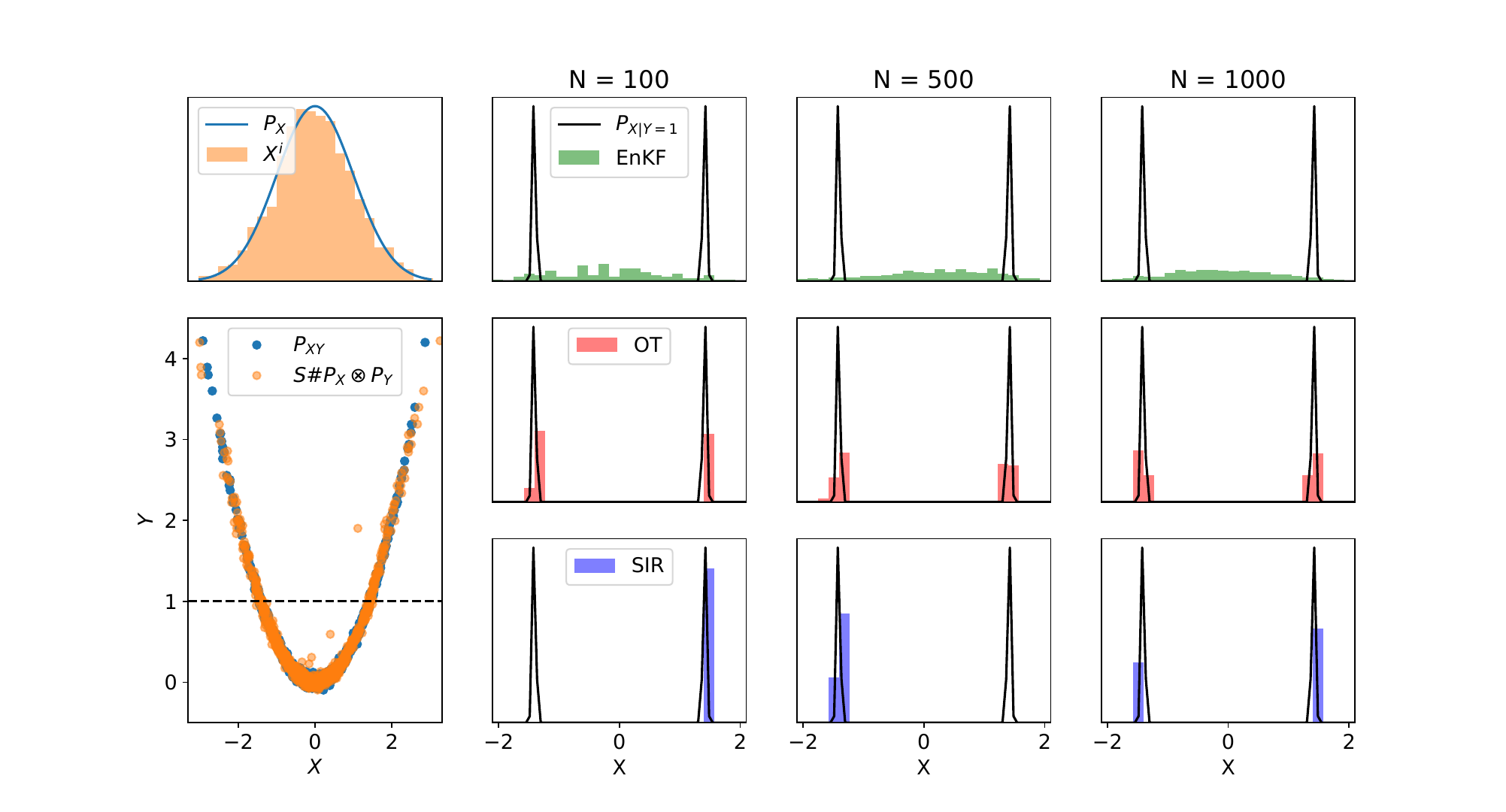} 
		\caption{$\lambda_w=0.04$.}
		\label{fig:squared_SNR}
	\end{subfigure}
	\caption{Numerical results for the static example in Sec.~\ref{sec:Static_Example}. (a) top-left: Samples $\{X^i\}_{i=1}^N$ from the prior $P_X$; bottom-left: samples $\{(X^i,Y^i)\}_{i=1}^N$ from the joint distribution $P_{X,Y}$ in comparison with the transported samples $\{(T(X^{\sigma_i},Y^i),Y^i)\}_{i=1}^N$; rest of the panels: transported samples for $Y=1$ for different values of $N$ and three different algorithms. (b) Similar results to panel (a) but  for a smaller $\lambda_w$.}
	\label{fig:squared}
\end{figure*}

\subsection{Relationship to the FPF Bayes update}\label{sec:FPF-OT}

	In this section, we present an alternative approach for derivation of the vector-fields $\k_t$ and $u_t$ to implement the FPF Bayes update formula.  The derivation is based on an application of the max-min optimization problem~\eqref{eq:new_loss}. 
	 For simplicity, the procedure is explained for identifying the vector-fields only at $t=0$. 
	To that end, consider a time discretization of
	the observation process~\eqref{eq:model-obs-cont} according to
	\begin{equation}\label{eq:obs-FPF}
		\Delta Y:= Y_{\Delta t} -Y_{0}= h(X_0) \Delta t+  W_{\Delta t}.
	\end{equation}
        where $X_0\sim \pi_0$ where it is assumed that $\pi_0$ admits a pdf denoted as $p_0$.  
	The state $X$ and the observation $\Delta Y$ are used to define the max-min problem~\eqref{eq:new_loss}.  
	In the limit of $\Delta t \to 0$, its solution is assumed to be of the form
	\begin{equation}\label{eq:f-FPF}
		f(x; y
		) = \phi(x) y+ \psi(x)\Delta t,\quad T(x,y) = x + \k(x)y + u(x)\Delta t
	\end{equation}
	where $\phi:\mathbb R^n \to \mathbb R$, $\psi:\mathbb R^n \to \mathbb R$, $u:\mathbb R^n \to \mathbb R^n$, and $\k:\mathbb R^n \to \mathbb R^n$ are functions that need to be determined.  
	The following proposition identifies the first-order and second-order approximation of the objective function in the asymptotic  limit as $\Delta t \to 0$. 
	
	\begin{proposition}\label{prop:FPF}
		Consider the objective function~\eqref{eq:new_loss}, with observation model~\eqref{eq:obs-FPF} and $f$ and $T$ specified as~\eqref{eq:f-FPF}. Then, in the asymptotic  limit as $\Delta t \to 0$, 
		\begin{align}\label{eq:J-FPF}
		J(f,T,P_{X,\Delta Y})&= J_1(\phi,\k)  \Delta t + J_2(\phi,\psi,\k,u)\Delta t^2 + O(\Delta t^3) 
	\end{align}
where
\begin{align*}
	J_1(\phi,\k) &:= \mathbb E[\frac{1}{2}|\k(X) - \nabla \phi(X)|^2 - \frac{1}{2}|\nabla \phi(X)|^2  + \phi(X)(h(X)-\hat h_0)] \\
	J_2(\phi,\psi,\k,u) &:= \mathbb E[\frac{1}{2} |u(X) |^2-\nabla \psi(X)u(X)  + u(X)^\top (\k(X)-\nabla \phi(X))\hat h_0 -\nabla \psi(X)^\top \k(X)\hat h_0 - \frac{1}{2}\k(X)^\top \nabla^2 \psi(X)\k(X) - \frac{3}{2} \k(X)^\top \nabla^2 \phi(X) \k(X)\hat h_0 ]
\end{align*}
where $\hat{h}_0 = \mathbb E[h(X)] = \int h(x) p_0(x) \ud x$. 
	\end{proposition}
\begin{proof}
See Appendix. 
\end{proof}

	The expansion of the objective function in~\eqref{eq:J-FPF} suggests that, in the limit as $\Delta t \to 0$,  the function $\phi$ and the vector-field $\k$ are obtained by solving the max-min problem
	\begin{align*}
		\max_\phi\min_\k\,J_1(\phi,\k)
	\end{align*}
Minimizing over  $\k(x)$, for a fixed function $\phi(x)$,   yields the solution $\k(x) = \nabla \phi(x)$, concluding the following maximization over $\phi$:
\begin{align*}
	\max_\phi\, \mathbb E[-\frac{1}{2}|\nabla \phi(X)|^2 + \phi(X)(h(X)-\hat{h}_0)] 
\end{align*}
The first-order optimality condition for $\phi$ concludes the Poisson equation 
\begin{equation}
	-\frac{1}{p_0(x)} \nabla \cdot (p_0(x)\nabla \phi(x)) = (h(x) - \hat{h}_0),\quad \forall x \in \Re^n,
\end{equation}
As a result, we recover the choice~\eqref{eq:K-u} made in the FPF algorithm.

The functions $\psi$ and $u$ are obtained by solving the max-min problem  for the second-order term $J_2$ in the expansion of the objective function. In particular, assuming the optimal form of the functions $\phi$ and $\k$, we have the max-min problem 
 	\begin{align*}
 	\max_\psi\min_u\, \mathbb E[\frac{1}{2}|u(X)|^2-\nabla \psi(X)u(X)  -\nabla \psi(X)^\top \nabla \phi(X)\hat h_0 - \frac{1}{2}\nabla \phi(X)^\top \nabla^2 \psi(X)\nabla \phi(X)  ]
 \end{align*}
 Minimizing over  $u(x)$, for a fixed function $\psi(x)$,   yields the solution $u(x) = \nabla \psi(x)$, concluding the following maximization over $\psi$:
 	\begin{align*}
	\max_\psi\, \mathbb E[-\frac{1}{2}|\nabla \psi(X)|^2 -\nabla \psi(X)^\top \nabla \phi(X)\hat h_0 - \frac{1}{2}\nabla \phi(X)^\top \nabla^2 \psi(X)\nabla \phi(X)  ]. 
\end{align*}
Upon using integration by parts for the last term, and the fact that $\phi$ solves the Poisson equation, the maximization over $\psi$ is expressed as
 	\begin{align*}
	\max_\psi\, \mathbb E[-\frac{1}{2} |\nabla \psi(X)|^2  + \nabla \psi(X)^\top v(X) ]
\end{align*}
where $v(x)= -\frac{h(x)+\hat{h}_0}{2} \nabla \phi(x) + \frac{1}{2}\nabla^2 \phi(x) \nabla \phi(x)$. 
The maximizer takes the form $\nabla \psi = v + \xi$
where $\xi$ is a divergence-free vector-field , i.e. $\nabla \cdot(p(x)\xi(x))=0$. 
	This is in agreement with the optimal transport form of the FPF algorithm proposed in~\citep{taghvaei2021optimal} and the original form of the FPF in~\citep{yang2016} modulo the additional divergence-free term $\xi$.  As explained in~\citep[Sidebar on pp.~40-41]{taghvaei2021optimal}, the addition of this term does not affect the evolution of the density.  Therefore, it may be chosen to be zero which is the choice made in the FPF.

\section{Optimal Transport Filter}\label{sec:OT_PF}

Implemented in a recursive manner, the max-min optimization formulation yields the OT filter as described in this section.

\subsection{OT filter design}
The OT filter is designed in three steps. 
\subsubparagraph{Step1: Exact mean-field process.}
We use the OT characterization of the conditional distribution to construct a (exact) mean-field process $\bar X_t$ whose distribution $\bar \pi_t$ is exactly equal to the posterior distribution $\pi_t$. The update equations for $\bar X_t$ are given by  
	\begin{equation}\label{eq:Xbar}
		\begin{aligned}
			\bar{X}_{t+1|t} & \sim a(\cdot|\bar X_t) ,\quad \bar X_{t+1} = \bar T_t(\bar X_{t+1|t},Y_{t+1})  \\
			\bar T_t &\leftarrow 	\max_{ f \in c\text{-concave}_x}\,
			\min_{T}\, J(f,T;P_{\bar X_{t+1|t}, \bar Y_{t+1|t}}),
		\end{aligned}
	\end{equation}
where $\bar Y_{t+1|t} \sim h(\cdot|\bar X_{t+1|t})$. It is then straightforward to verify that
\begin{equation}
	\begin{aligned}
		\bar \pi_{t+1} &= \bar T_t(\cdot,Y_{t+1})_{\#} \mathcal A  \bar \pi_{t}
		= \mathcal B_y (\mathcal A \bar \pi_{t}),
	\end{aligned}
	\label{eq:exact-mean-field-dist}
\end{equation}
where the second {identity} is a consequence of Proposition~\ref{prop:consistency}. It then follows that
whenever $\bar \pi_0 = \pi_0$ then $\bar \pi_t = \pi_t$. As such, the mean-field process $\bar X_t$ is called exact.
The OT filter is obtained by approximating the exact mean-field process $\bar X_t$.

\subsubparagraph{Step 2: Approximate mean-field process.}
This step consists of restricting the feasible set of the optimization problem to a parameterized class of functions $\mathcal F$ and $\mathcal T$.
The resulting approximated process, denoted by $\tilde X_t$, follows the update rule:
	\begin{equation}
		\begin{aligned}
				\tilde{X}_{t+1|t} & \sim a(\cdot|\tilde X_t) ,\quad \tilde X_{t+1} = \tilde T_t(\tilde X_{t+1|t},Y_{t+1})
					\\
			\tilde T_t &\leftarrow 	\max_{ f \in \mathcal F}\,
			\min_{T \in \mathcal T}\, J(f,T;P_{\tilde X_{t+1|t}, \tilde Y_{t+1|t}}),
		\end{aligned}
		\label{eq:approx-mean-field-dist}
	\end{equation}
where $\tilde Y_{t+1|t} \sim h(\cdot|\tilde X_{t+1|t})$.
	This update defines the approximate mean-field distribution $\tilde \pi_t$ that follows  the update 
	\begin{align}\label{eq:XF}
\tilde \pi_{t+1|t} = \mathcal A \tilde \pi_t,\quad \tilde \pi_{t+1} =   \tilde T_t (\cdot,Y_{t+1})_{\#}\tilde \pi_{t}
	\end{align}

\subsubparagraph{Step 3: Finite particle system.}
The second approximation step is to replace the mean-field process with an empirical distribution of a collection of
particles $\{X^1_t,\ldots,X^N_t\}$, i.e.,  $\tilde \pi_t \approx \frac{1}{N}\sum_{i=1}^N \delta_{X_t^i}.$ It follows by discretization of the update equation~\eqref{eq:XF} for the mean-field process $\tilde X_t$ according to
\begin{equation}\label{eq:particles-interacting}
	\begin{aligned}
				{X}^i_{t+1|t} & \sim a(\cdot|X^i_t) ,\quad  X^i_{t+1} = \hat T_t(X^i_{t+1|t},Y_{t+1})
\\
\hat T_t & \leftarrow 	\max_{ f \in \mathcal F}\,
\min_{T \in \mathcal T}\,  J(f,T; \frac{1}{N}\sum_{i=1}^N \delta_{(X^i_{t+1|t},Y^i_{t+1|t})}),
	\end{aligned}
\end{equation}
where  $Y^i_{t+1|t} \sim h(\cdot|X^i_{t+1|t})$. 

\subsection{Error analysis}

For error analysis of the filter, the following assumption is needed.

\begin{definition}[Uniformly geometrically  stable filter]
	The filter update~\eqref{eq:exact-posterior} 
	is uniformly geometrically stable 
	if $\exists \lambda\in (0,1)$ and positive constant $C>0$ such that for all $\mu,\nu$ and $ t>s\geq 0$ it holds that
	\begin{equation}\label{eq:stability}
		d(\mathcal T_{t,s}\mu,\mathcal T_{t,s}\nu) \leq C(1-\lambda)^{t-s}d(\mu,\nu).
	\end{equation}
	\label{def:stability}
\end{definition}
\medskip

The distance between the exact mean-field distribution $\bar \pi_t$ and the approximate mean-field distribution $\tilde \pi_t$ is characterized in the following proposition.

\begin{proposition}\label{prop:mean-field}
	Consider $\bar \pi_t$ and $\tilde \pi_t$ as in~\eqref{eq:exact-mean-field-dist}-\eqref{eq:XF}{, respectively}. Assume 
	\begin{enumerate}
		\item The exact filter is stable according to Definition~\ref{def:stability}. 
		\item There  max-min optimality gap is uniformly bounded by  $\epsilon>0$. 
		\item 	For all $y$ and $t$, the function  $x\mapsto \frac{1}{2}|x|^2 - \tilde f_t( x,y)$ is $\alpha$-strongly convex 	\end{enumerate}
	Then, it holds that
	\begin{align}\label{eq:pitF-bound}
		d(\tilde \pi_t,\pi_t) \leq \frac{C}{\lambda}\sqrt{\frac{4\epsilon}{\alpha}},\quad \forall t,
	\end{align}
	with all constants independent of time. 
\end{proposition}
\begin{proof}The proof appears in~\cite[Prop. 2]{al2023optimal}.
\end{proof}

\begin{remark}[Relationship to literature] \label{rem:stability}
	The uniform geometric stability property~\eqref{eq:stability} is also used in the error analysis of PFs in~\cite{del2001stability,delmoralbook}. It can be verified if the dynamic transition kernel satisfies a minorization condition,
	i.e., there exists a probability measure $\rho$ 
	and a constant $\epsilon >0$ such that  $a(x|x')\geq \epsilon\rho(x)$.
	The minorization is a mixing condition that ensures geometric ergodicity of the Markov process $X_t$~\cite{meyn2012markov}. We acknowledge that this condition is  {strong} and can be verified for a restricted class of systems, e.g., $X_t$ should belong to a compact set. A complete characterization of  systems with uniform geometric stable filters is an open and challenging problem in the field. More insight is available for the weaker notion of asymptotic stability of the filter, i.e., $\lim_{t\to\infty}d(\mathcal T_{t,s}\mu,\mathcal T_{t,s}\nu)=0$, which holds when the system is  ``detectable''
	in a sense that is suitable for nonlinear stochastic dynamical systems~\cite{van2010nonlinear,chigansky2009intrinsic,van2009observability,kim2022duality}. This characterization of systems with asymptotic filter stability 
	is in agreement with the existing results for the stability of the Kalman filter, which holds when the linear system is detectable in the classical sense~\cite{ocone1996}.  A complete survey of existing  filter stability results can be found in~\cite{crisan2011oxford}. 
\end{remark}

\begin{remark}[Error analysis of the finite-$N$ system] 
The derivation of an error bound for  the particle system is challenging because the particles become correlated constituting an interactive particle system, requiring application of tools from  propagation of chaos~\cite{sznitman1991}). The error analysis of the interacting particle system is the subject of ongoing work. However, it is possible to provide an error bound for the particle system equipped with a resampling  stage so that particles become independent of each other. The error-bound is
similar to the mean-field analysis presented in the previous
proposition, with an additional error term $\frac{1}{\sqrt{N}}$ due to the sampling. See \cite[Prop. 3]{al2023optimal} for details.
\end{remark}

\subsection{Literature survey and comparison}

Quantifying uncertainty and effectively assimilating noisy sensory data is the subject of nonlinear filtering and is crucial for the reliable and safe operation of control systems. 
Classical nonlinear filtering algorithms, such as Kalman filter  with its nonlinear extensions~\cite{kalman1960new,kalman-bucy,bar2004estimation}, and particle filters (PF)~\cite{gordon1993novel,arulampalam2002tutorial,doucet09} are subject to fundamental limitations that prohibits their application to modern high-dimensional problems with strong nonlinear effects: Kalman filters are sensitive to initial conditions and fail to represent multi-modal distributions~\cite{gordon2004,budhiraja2007survey}; PF suffer from the  particle degeneracy phenomenon which becomes severe in high-dimensional problems, an issue known as the curse of dimensionality~\cite{bickel2008sharp,rebeschini2015can,beskos2014error,bengtsson08}.

These issues motivated recent efforts in the nonlinear filtering literature to develop numerical algorithms
based on a controlled system of interacting particles to approximate the posterior
distribution~\cite{Tao_TAC,yang2016,crisan10,reich11,reich2015probabilistic,bergemann2012ensemble,daum10,daum2017generalized}. A prominent idea is to view the problem of transforming samples from the prior to
the posterior from the lens of optimal transportation theory \cite{reich2013nonparametric,reich2019data,taghvaei2020optimal,reich13,AmirACC2016,taghvaei2021optimal}, which has
also become popular in the Bayesian inference literature~\cite{el2012bayesian,marzouk2016introduction,
	mesa2019distributed,heng2015gibbs,kovachki2020conditional,siahkoohi2021preconditioned}. See ~\cite{spantini2022coupling} and \cite{taghvaei2023survey} for a recent survey of these topics. 
{
	Broadly speaking, the aim of the above methods is to find a  {transport map} (be it stochastic or deterministic)
	that transforms the prior distribution to the posterior distribution while minimizing a certain cost.

\subsubparagraph{Comparison of the OT approach with other coupling-based methods.} The particle flow method~\cite{daum2012particle,de2015stochastic} and feedback particle filter (FPF)~\cite{Tao_TAC,yang2016} involve either  an ordinary differential equation or stochastic differential equation that updates the locations of the particles so that the probability density of the particles follows a given PDE. The particles' equation involves an unknown vector-field that needs to be approximated by solving a certain partial differential equation (PDE). 
The main challenge in this type of algorithms is to approximate the aforementioned vector-field at each time-step. The time discretization
for this type of equation often becomes unstable, especially for multi-modal posteriors or degenerate likelihoods. Our OT approach can be viewed as an exact time-discretization of the FPF algorithm, as shown in Section~\ref{sec:FPF-OT}, which resolves the time-discretization issues discussed above. 

The ensemble transform particle filter~\cite{reich11}
involves solving a linear program for the discrete OT problem from a uniform prior distribution to the weighted posterior distribution for the particular value of the observation.   Solving the linear program becomes challenging as the number of particles $N$ increases. Moreover, approximating the marginal with a weighted empirical distribution suffers from the same fundamental issue that importance sampling particle filter suffers from.  

The coupling method proposed in~\cite{spantini2022coupling} is the closest method to our approach. It is also likelihood free and amenable to the neural net parameterizations.  The main difference is in the form of the transport map. While in this paper we aim at finding the OT map from prior to the posterior, the approach in~\cite{spantini2022coupling} aims at finding the Knothe–Rosenblatt rearrangement. Thus our approach is more closely related to the semi-dual solutions to the OT problem and can utilize the existing theoretical results and computational methodologies.

\section{Numerics}\label{sec:numerics}
We illustrate the performance of the OT filter, in comparison with the EnKF and SIR PF, for two nonlinear filtering examples.
\subsubsection{A toy example with bimodal posterior}
We consider a dynamic version of the static example in Sec.~\ref{sec:Static_Example} according to the following model:
\begin{subequations}\label{eq:model-example}
	\begin{align}
		X_{t} &= (1-\alpha) X_{t-1} + 2\lambda V_t,\quad X_0 \sim \mathcal{N}(0,I_n)\\
		Y_t &= X_t\odot X_t + \lambda W_t,
	\end{align}
\end{subequations}
where $\{V_t,W_t\}_{t=1}^\infty$ are i.i.d sequences of standard Gaussian random variables, $\alpha=0.1$ and $\lambda=\sqrt{0.1}$. The choice of $Y_t$ 
will once again lead to a bimodal posterior $\pi_t$ at every time step.

The numerical results are depicted in Fig.~\ref{fig:dynamic_example_states}: Panel (a) shows the trajectory of the particles for the three algorithms, along with the true state $X_t$ denoted with a dashed black line. The OT approach produces a bimodal distribution of particles, while the EnKF gives a Gaussian approximation and the SIR approach exhibits the weight collapse and misses a mode for the time duration $t\in[1,2] \cup [3,4.5]$. Panel (b) presents a quantitative error analysis comparing the maximum-mean-discrepancy (MMD) between the particle distribution of each algorithm  and the exact posterior.
Since the exact posterior is not explicitly available it is approximated by simulating the SIR algorithm with $N=10^5$ particles.  This quantitative result affirms the qualitative observations of panel (a) that the OT posterior better captures the true posterior in time.

We also performed a numerical experiment to study the effect of the dimension $n$ and the number of particles $N$ on the performance of the three algorithms. The results are depicted in panels (c) and (d), respectively. It is observed that both EnKF and OT scale better with dimension compared to SIR. However, as the number of particles increases, the EnKF error remains constant, due to its Gaussian bias, while the approximation error for SIR and OT decreases. 

\begin{figure*}[t]
	\centering
	\begin{subfigure}{0.24\textwidth}
		\centering
		\includegraphics[width=1.0\textwidth,trim={30 0 70 60},clip]{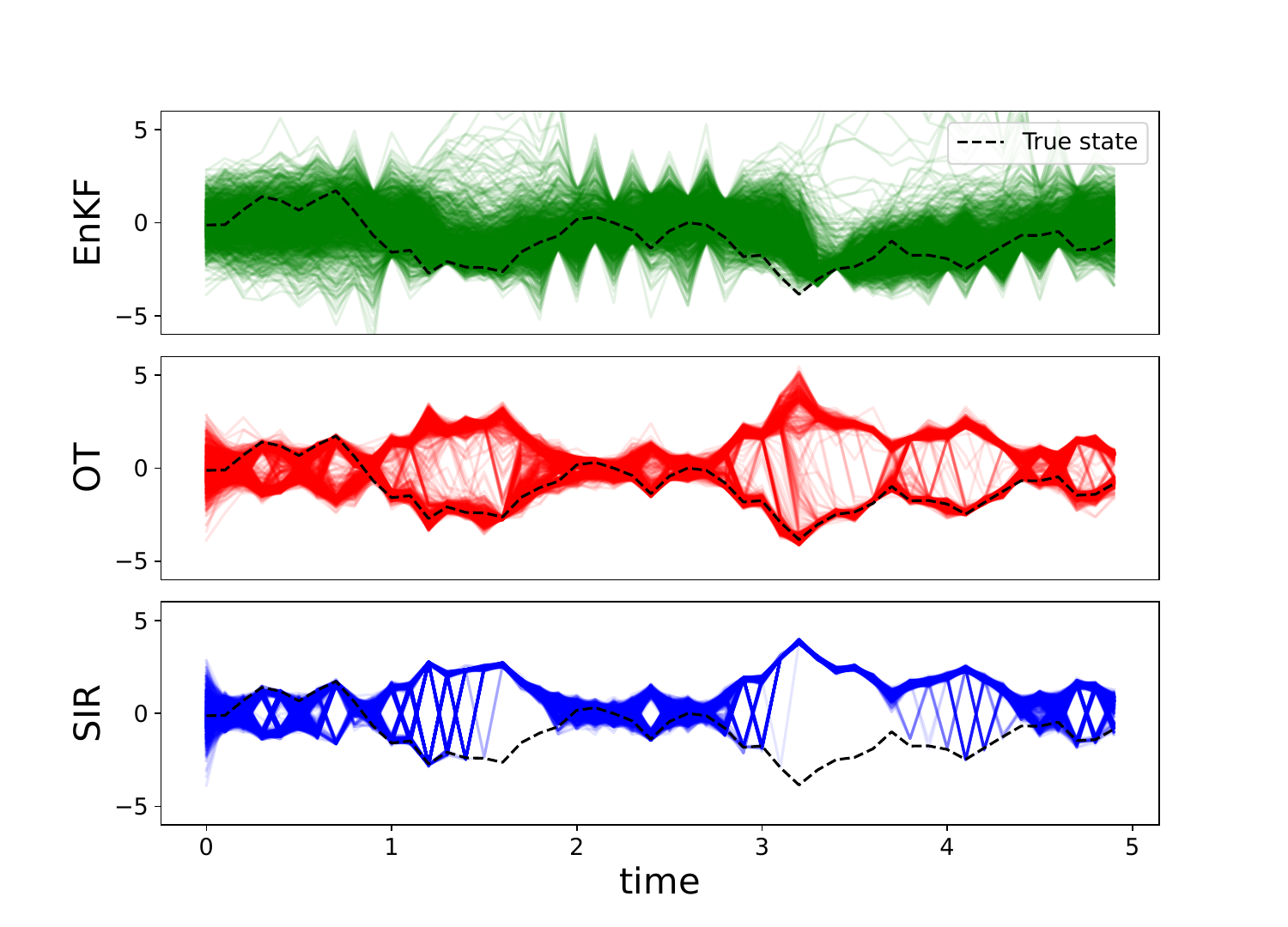}
		\caption{Particles trajectory.}
	\end{subfigure}
	\hfill
	\begin{subfigure}{0.24\textwidth}
		\centering
		\includegraphics[width=1.0\textwidth,trim={30 0 70 60},clip]{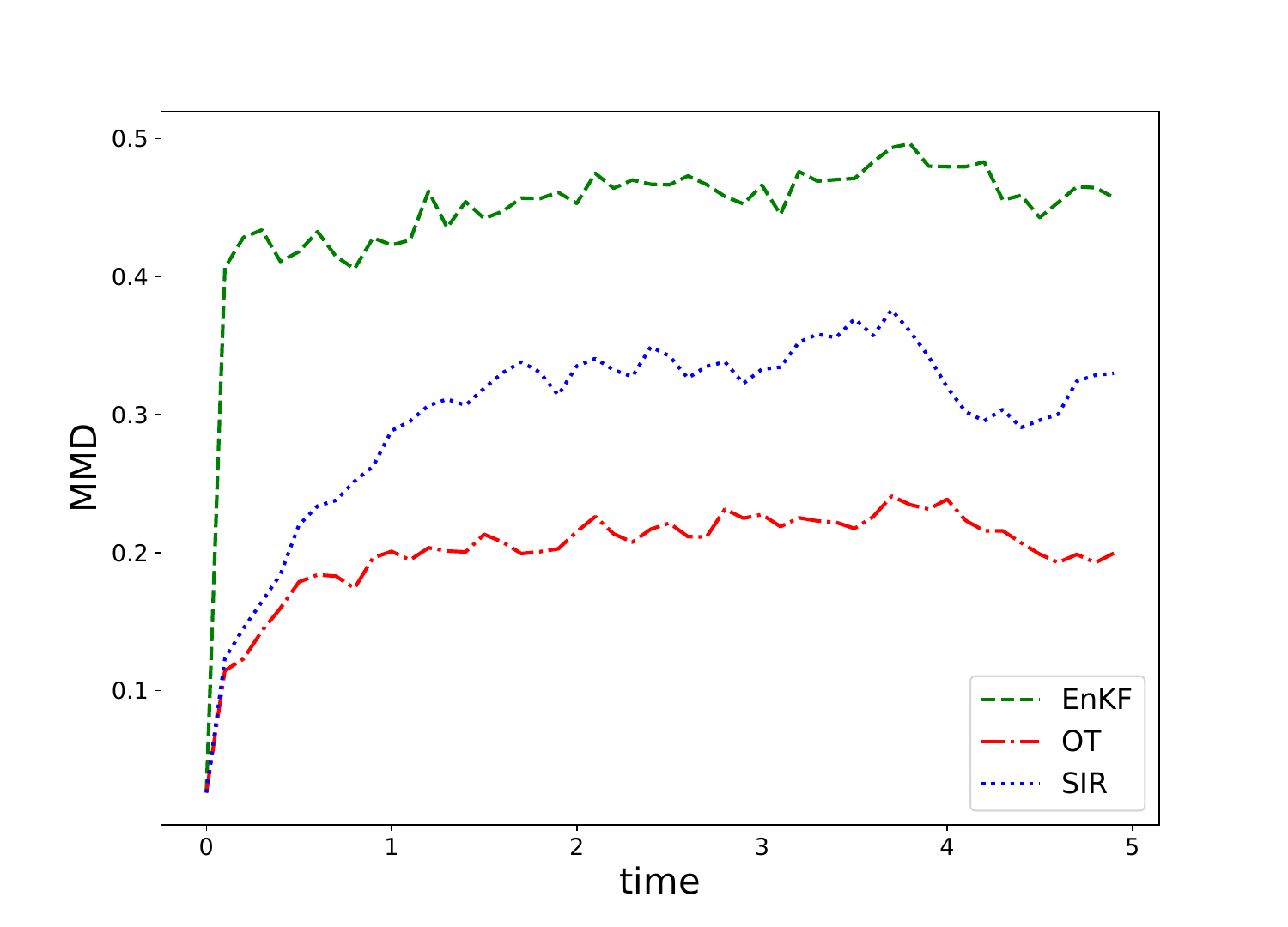} 
		\caption{MMD vs time.}
	\end{subfigure}
	\hfill
	\begin{subfigure}{0.24\textwidth}
		\centering
		\includegraphics[width=1.0\textwidth,trim={30 0 70 60},clip]{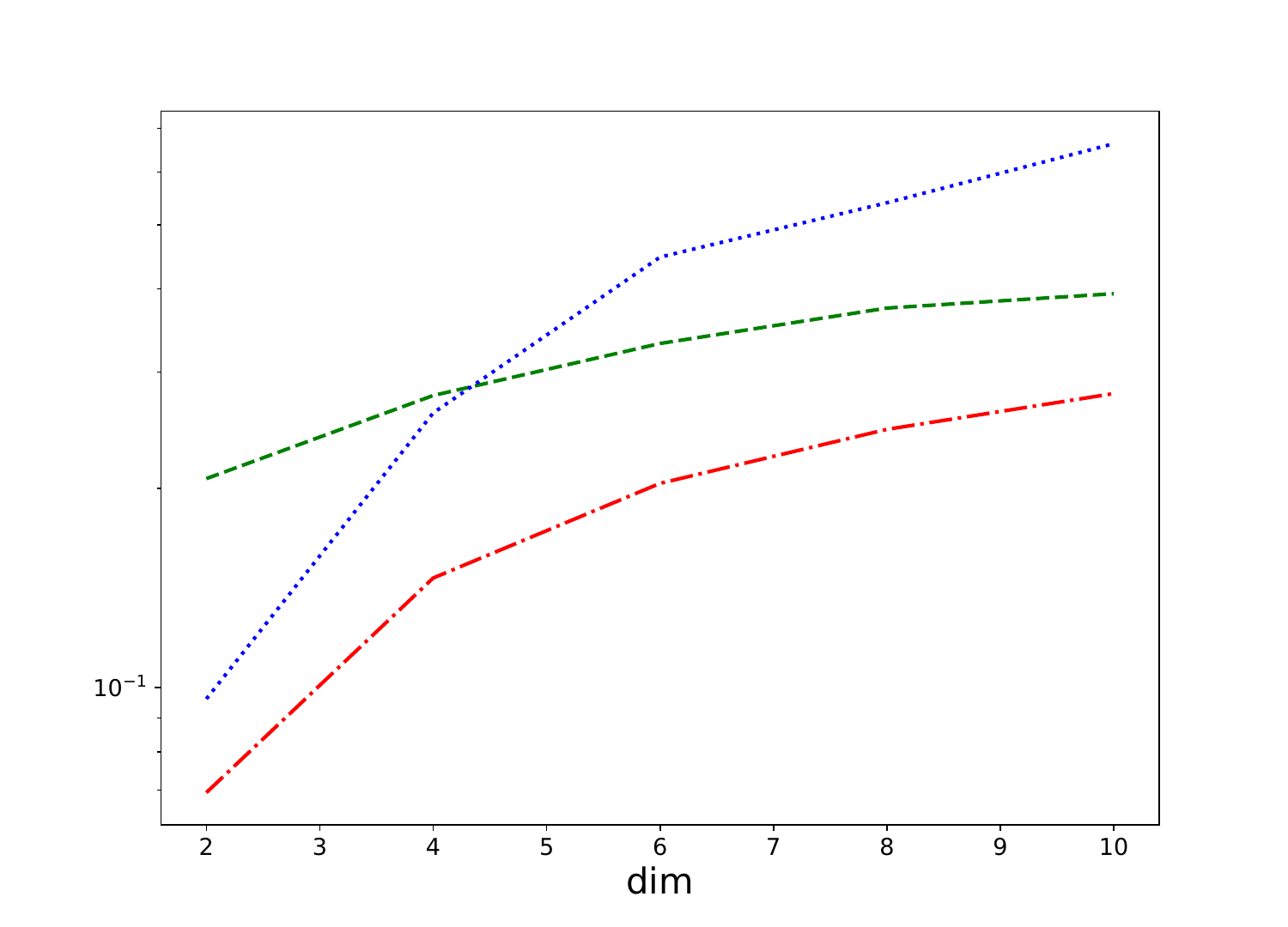}
		\caption{MMD vs the dimension.}
	\end{subfigure}
	\hfill
	\begin{subfigure}{0.24\textwidth}
		\centering
		\includegraphics[width=1.0\textwidth,trim={30 0 70 60},clip]{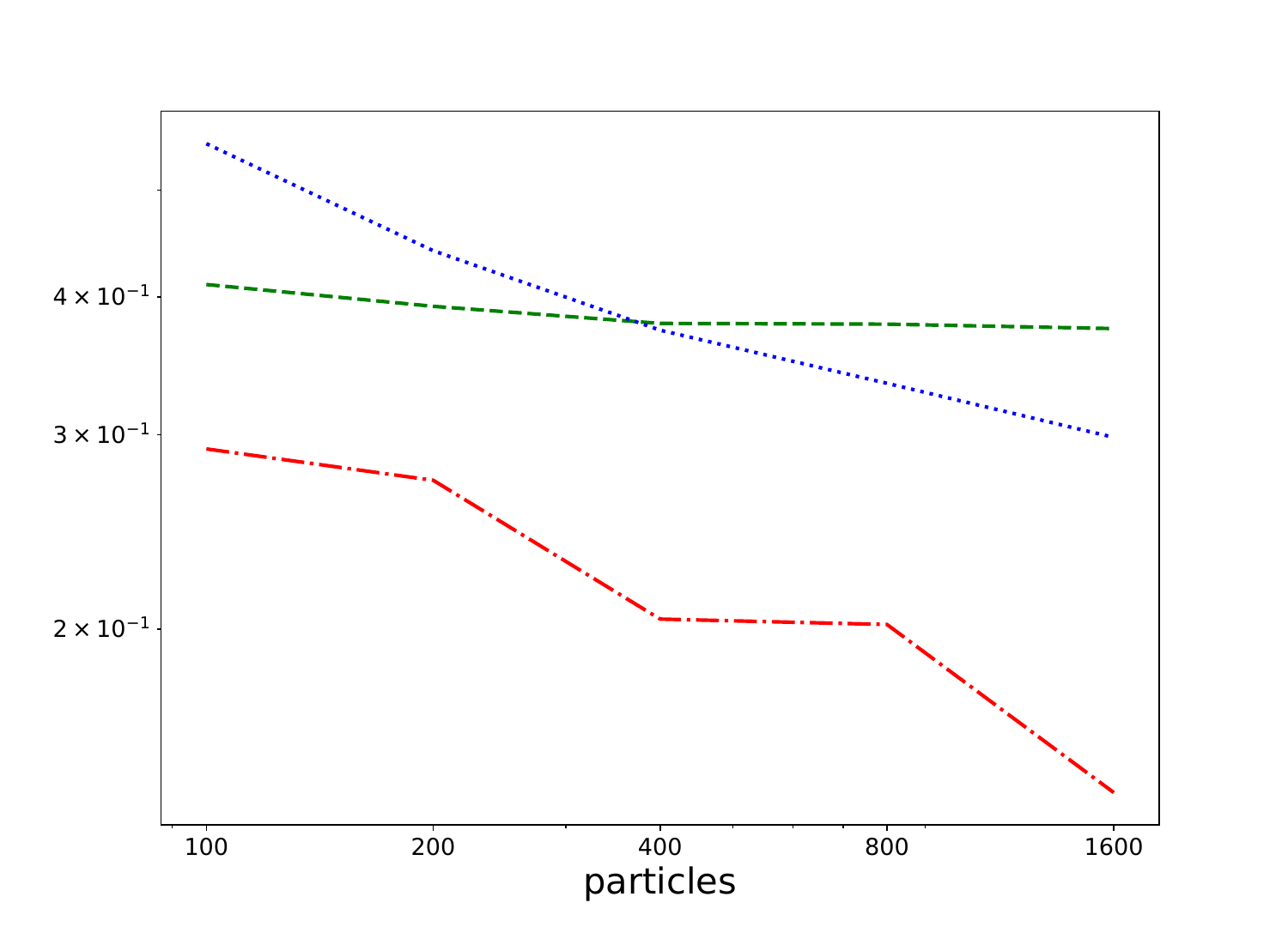} 
		\caption{MMD vs $\#$ of particles.}
	\end{subfigure}
	
	\caption{Numerical results for the dynamic example~\ref{eq:model-example}. The left panel shows the trajectory of the particles $\{X^1_t,\ldots,X^N_t\}$ along with the trajectory of the true state $X_t$ for EnKF, OT, and SIR algorithms, respectively. The second panel shows the MMD distance with respect to the exact conditional distribution. The last two panels show MMD variation with dimension and the number of particles.}
	\label{fig:dynamic_example_states}
\end{figure*}
\subsubsection{Lorentz-63}
\begin{figure}[t]
	\centering
	\includegraphics[width=0.43\textwidth,trim={75 0 100 20},clip]{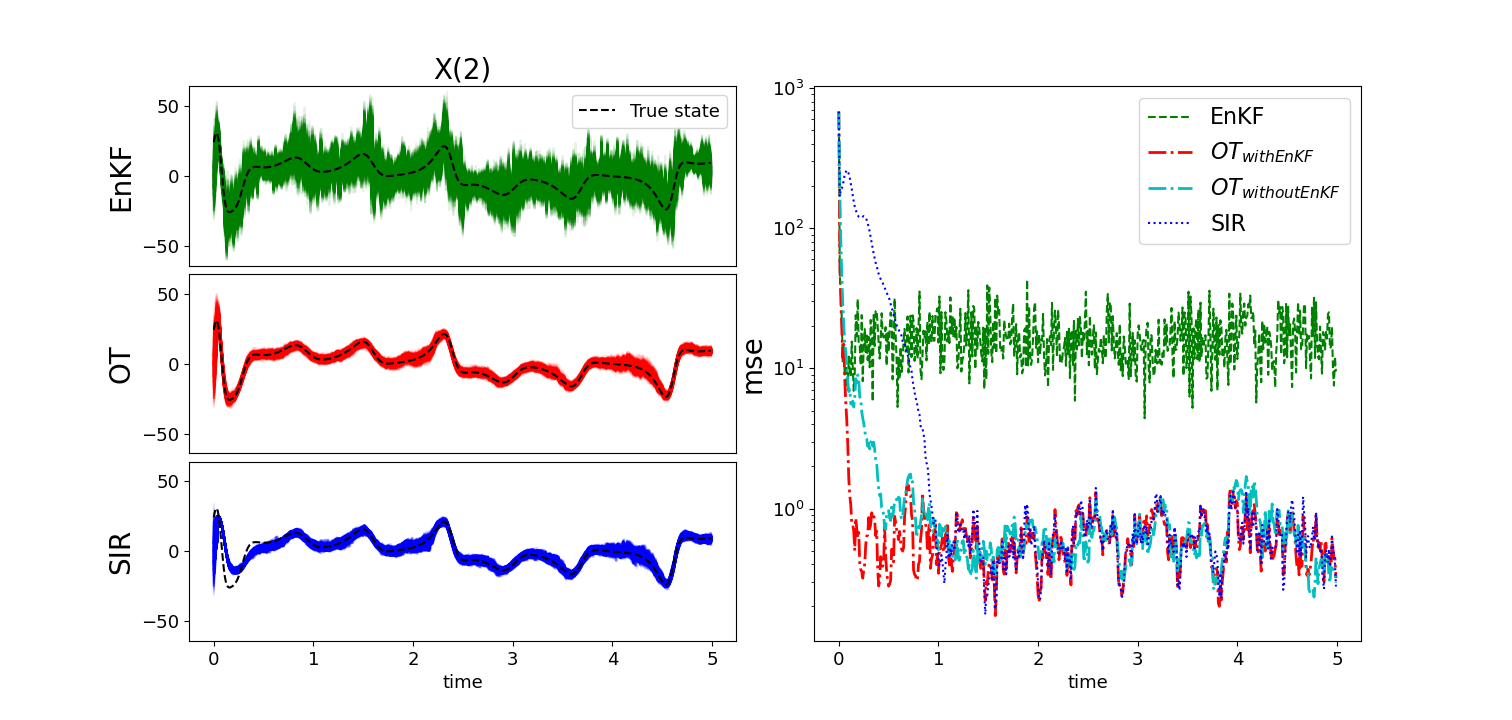}
	\caption{Numerical results for the Lorenz 63 example. The left panel shows the trajectory of the unobserved component of the true state and the particles. The right panel shows the MSE comparison.}
	\label{fig:state2_and_mse_L63}
\end{figure}
We present numerical results on the three-dimensional  Lorenz 63 model which often serves as a benchmark for nonlinear filtering algorithms. The model details appear~\cite[Appendix C.4]{al2023high}. The
state $X_t$ is $3$-dimensional while the observation $Y_t$ is $2$-dimensional and consists of noisy measurements of the first and third components of the state.

The numerical results are presented in Fig.~\ref{fig:state2_and_mse_L63}. The left panel shows the trajectory of the second component of the true state and the particles.  The OT and EnKF are quicker in converging to the true state, with EnKF admitting larger variance. The right panel shows the mean-squared-error (MSE) in estimating the state confirming the qualitative observations. We present the MSE result for two variations of the OT method: either the EnKF layer in the architecture of Fig.~\ref{tikz:static_struc} is implemented or not.  The results show that the addition of the EnKF layer helps with the performance of the filter, while computationally, we observed more numerical stability when the EnKF layer is removed.

\section{Conclusions}\label{sec:conc}
We presented a summary of the recent development in optimal transportation methods for  nonlinear  filtering problems. Specifically, we focused on the OT formulation of the Bayes' law, which led to a simulation-based nonlinear filtering algorithm that is able to capture multi-modal posterior distributions. Additional experiments in~\cite{al2023high} highlights the scalability of the approach to high-dimensional settings involving images, while it is noted that the raw 
computational time of the OT approach is higher, and 
for nonlinear filtering examples that admit unimodal posterior, such as Lorentz-96, the EnKF provides a fast and reasonable approximation. 
A computational feature of the OT method is that it provides the user with the flexibility to set the computational budget: without any training, OT algorithm implements EnKF; with additional budget (increasing training iterations and complexity of the neural net), the accuracy is increased, see Appendix C.1 in~\cite{al2023high}. The computational efficiency of the OT approach can be improved by fine-tuning the neural network architectures, optimizing the hyper-parameters, and including an offline training stage for the first time step, which will be used as a warm-start for training at future time steps in the online implementation. This line of research is pursued in our recent work~\cite{al2024data}, where a new data-driven nonlinear filtering algorithm was introduced aimed at ergodic state and observation dynamics. The algorithm consists of offline and online stages: The offline stage is expensive to train and learns a static conditioning transport map; The online stage is computationally cheap and uses the learned conditioning map without any further training, providing a competitive computational time compared with traditional methods during online inference.

\begin{ack}[Acknowledgments]

\end{ack}


\bibliographystyle{Harvard}
\bibliography{reference}

\appendix
\section{Calculations for proof of Prop.~\ref{Prop:EnKF}}
\begin{align*}
	\mathbb E[\overline{X}_{\mid_y} ] &=  \mathbb E [\overline X] + K(y-\mathbb E [\overline Y]) = m + K(y-Hm),\\ \text{Cov}(\overline{X}_{\mid_y})&= \text{Cov}(\overline X)  + K \text{Cov}(\overline Y) K\tp - \text{Cov}(\overline X,\overline Y)K\tp - K \text{Cov}(\overline X,\overline Y)= \Sigma - \Sigma H\tp (H\Sigma H\tp + R)^{-1} H \Sigma 
\end{align*}
\section{Proof of Prop.~\ref{prop:EnKF-error}}
	For any function $g\in \mathcal G$: 
\begin{align*}
	\int g(x) P_{X|Y}^{\text{EnKF}}(x|y)  \ud x &= \frac{1}{N}\sum_{i=1}^N g(X^i_0 +\hat K(y-Y^i_0))\\
	\int g(x) P_{X|Y}(x|y)  \ud x &= \mathbb E[g(X_0 +K(y-Y_0))]
\end{align*}
As a result, 
\begin{align*}
	\mathbb E&[|\int g(x) P_{X|Y}^{\text{EnKF}}(x|Y)  \ud x  -  	\int g(x) P_{X|Y}(x|Y)  \ud x |^2]^{\frac{1}{2}} =\mathbb E[| \frac{1}{N}\sum_{i=1}^N g(X^i_0 +\hat K(Y-Y^i_0)) - \mathbb Eg(X_0 + K(Y-Y_0)) |^2]^{\frac{1}{2}} \\
	&\leq \mathbb E[| \frac{1}{N}\sum_{i=1}^N g(X^i_0 +K(Y-Y^i_0)) - \mathbb Eg(X_0 + K(Y-Y_0)) |^2]^{\frac{1}{2}} + \mathbb E[| \frac{1}{N}\sum_{i=1}^N g(X^i_0 +\hat K(Y-Y^i_0)) -g(X^i_0 + K(Y-Y^i_0)) |^2]^{\frac{1}{2}}\\
	&\leq \frac{\text{Var}(g(X_0 + K(Y-Y_0)))}{\sqrt{N}} +  \mathbb E[\|(\hat K - K)(Y-Y_0)\|^2]^{\frac{1}{2}}
\end{align*}
where, in order to derive the first term, we used the fact that  $(X^i_0,Y^i_0)$ are i.i.d. samples of $P_{X,Y}$, and, in order to derive the second term, we used the fact that the function $g \in \mathcal G$ is Lipschitz with constant $1$. Finally, the fact that $g$ is uniformly bounded by $1$ and Cauchy-Schwartz inequality concludes the bound~\eqref{eq:EnKF-bound}. 
\section{Proof of Prop.~\ref{prop:FPF}}
	The proof follows from 
the expansion of the three terms in the objective function~\eqref{eq:new_loss}, up to second-order in $\Delta t$: 
\begin{align*}
	\mathbb E[f(X,\Delta Y)] &=  \mathbb E[\phi(X)\Delta Y] +\mathbb E[\psi(X)]\Delta t\\
	\mathbb E[f(T(\overline X,\Delta Y),\Delta Y)] &= \mathbb E[\phi(T(\overline X,\Delta Y))\Delta Y] + \mathbb E[\psi(T(\overline X,\Delta Y))]\Delta t \\&=  \mathbb E[\phi(\overline X + \k(\overline X) \Delta Y + u(\overline X) \Delta t)\Delta Y] + \mathbb E[\psi(\overline X + \k(\overline X) \Delta Y + u(\overline X) \Delta t)]\Delta t \\
	&= \mathbb E[\phi(\overline X)\Delta Y] + \mathbb E[\nabla \phi(\overline X)^\top K(\overline X) (\Delta Y)^2] +  \mathbb E[\nabla \phi(\overline X)^\top u(\overline X) \Delta t\Delta Y] + \frac{1}{2} \mathbb E[\k(\overline X)^\top\nabla^2 \phi(\overline X) \k(\overline X) (\Delta Y)^3]  \\&+ \mathbb E[\psi(\overline X)\Delta t] +  \mathbb E[\nabla \psi(\overline X)^\top \k(\overline X) \Delta t\Delta Y] +\mathbb E[\nabla \psi(\overline X)^\top u(\overline X) (\Delta t)^2] + \frac{1}{2} \mathbb E[\k(\overline X)^\top\nabla^2 \psi(\overline X) \k(\overline X) (\Delta Y)^2\Delta t] + O(\Delta t^3)\\
	\frac{1}{2}\mathbb  E[\|T(\overline X,\Delta Y)-\overline X\|^2]&= \frac{1}{2}\mathbb  E[\|\k(\overline X)\|^2 (\Delta Y)^2] + \frac{1}{2}\mathbb  E[\|u(\overline X)\|^2 (\Delta t)^2] +  E[\k(\overline X)^\top u(\overline X) \Delta t\Delta Y],
\end{align*}
and using the relationship $\Delta Y = h(X) \Delta t+ \Delta W$. 
\end{document}